\documentclass[a4paper, onecolumn, superscriptaddress, longbibliography, accepted=2023-07-17]{quantumarticle}
\pdfoutput=1
\usepackage[utf8]{inputenc}
\usepackage{amsmath,amsthm}
\usepackage{amssymb}
\usepackage{mathtools}
\usepackage{dsfont}
\usepackage{xcolor}
\usepackage{algorithm,algpseudocode}
\usepackage{graphicx}
\usepackage{subcaption}
\usepackage[colorlinks=true, linkcolor=blue, citecolor=blue, urlcolor=blue]{hyperref}
\usepackage{cleveref}
\usepackage{nicefrac}
\newcommand{\ts}{\textsuperscript}
\usepackage{thm-restate}
\usepackage{mathrsfs}
\usepackage{qcircuit}

\newtheorem{theorem}{Theorem}
\newtheorem{lemma}[theorem]{Lemma}
\newtheorem{claim}[theorem]{Claim}

\usepackage{todonotes}

\newcommand{\capU}{\mathcal{U}}
\newcommand{\capUA}{\mathcal{U}_A}
\newcommand{\capUB}{\mathcal{U}_B}
\newcommand{\curlyP}[1]{{\mathcal{U}_{TS}^{( #1 )}}}
\newcommand{\trotterU}[1]{{U_{TS}^{( #1 )}}}
\newcommand{\qdchan}{\mathcal{U}_{QD}}
\newcommand{\hilbSpace}{\mathscr{H}}

\newcommand{\bigo}[1]{O\left( #1 \right)}
\newcommand{\bigotilde}[1]{\widetilde{O} \left( #1 \right)}
\newcommand{\probIndexSet}{\mathcal{S}}

\newcommand{\prob}[1]{\text{Pr}\left[ #1 \right]}

\newtheorem{definition}[theorem]{Definition}
\newcommand{\ket}[1]{|#1\rangle}
\newcommand{\bra}[1]{\langle #1|}

\newcommand{\ketbra}[2]{| #1\rangle\! \langle #2|}
\newcommand{\parens}[1]{\left( #1 \right)}
\newcommand{\brackets}[1]{\left[ #1 \right]}
\newcommand{\abs}[1]{\left| #1 \right|}
\newcommand{\norm}[1]{\left| \left| #1 \right| \right|}
\newcommand{\diamondnorm}[1]{\left| \left| #1 \right| \right|_\diamond}

\newcommand{\set}[1]{\left\{ #1 \right\}}
\newcommand{\ceil}[1]{\left\lceil #1 \right\rceil}
\newcommand{\expect}[1]{\mathbb{E}\brackets{#1}}

\begin{document}
\title{Composite Quantum Simulations}
\author{Matthew Hagan}
\affiliation{Department of Physics, University of Toronto, Toronto ON, Canada}
\author{Nathan Wiebe}
\affiliation{Department of Computer Science, University of Toronto, Toronto ON, Canada}
\affiliation{Pacific Northwest National Laboratory, Richland Wa, USA}
\affiliation{Canadian Institute for Advanced Study, Toronto ON, Canada}

\begin{abstract}
    In this paper we provide a framework for combining multiple quantum simulation methods, such as Trotter-Suzuki formulas and QDrift into a single Composite channel that builds upon older coalescing ideas for reducing gate counts. The central idea behind our approach is to use a partitioning scheme that allocates a Hamiltonian term to the Trotter or QDrift part of a channel within the simulation.  This allows us to simulate small but numerous terms using QDrift while simulating the larger terms using a high-order Trotter-Suzuki formula.  We prove rigorous bounds on the diamond distance between the Composite channel and the ideal simulation channel and show under what conditions the cost of implementing the Composite channel is asymptotically upper bounded by the methods that comprise it for both probabilistic partitioning of terms and deterministic partitioning.  Finally, we discuss strategies for determining partitioning schemes as well as methods for incorporating different simulation methods within the same framework.
\end{abstract}
\maketitle

\section{Introduction}\label{sec:intro}
The simulation of quantum systems remains one of the most compelling applications for future digital quantum computers~\cite{whitfield2011simulation,jordan2012quantum,reiher2017elucidating,babbush2019quantum,su2021fault,o2021efficient}. As such, there are a plethora of algorithm options for compiling a unitary evolution operator $U(t) = e^{-i H t}$ to circuit gates~\cite{aharonov2003adiabatic,berry2007efficient,berry2015simulating,childs2019faster,low2019hamiltonian,low2019well,low2018hamiltonian, qdrift}. Some of the simplest such algorithms are product formulas in which each term in a Hamiltonian $H = \sum_i h_i H_i$ is implemented as $e^{iH_it}$. A product formula is then a particular sequence
of these gates that approximates the overall operator $U(t)$. Two of the most well known product formula include Trotter-Suzuki Formulas~\cite{berry2007efficient,wiebe2010higher,childs2019faster,childs2021theory} and the QDrift protocol in which terms are sampled randomly~\cite{qdrift,berry2020time}. These two approaches are perhaps the most popular ancilla-free simulation methods yet discovered.  

One of the main drawbacks of Trotter-Suzuki formulas is that each term in the Hamiltonian has to be included in the product formula regardless of the magnitude of the term.  This leads
to a circuit with a depth that scales at least linearly with the number of terms in $H$, typically denoted $L$. QDrift avoids this by randomly choosing which
term to implement next in the product formula according to an importance sampling scheme in which higher weight terms have larger probabilities. The
downside to QDrift is that it has the same asymptotic scaling with $t/\epsilon$ as a first-order Trotter formula, meaning it is outperformed at large
$t/\epsilon$ by even a second-order Trotter formula. 


In this paper we present a framework for combining simulation channels in a way that allows one to flexibly interpolate the gate cost tradeoffs between the individual channels. The primary example we study is the composition of Trotter-Suzuki and QDrift channels. This is motivated in some part as an effort to extend 
randomized compilers to include conditional probabilities and in some part to encapsulate progress in chemistry simulations of dropping small
weight terms or shuffling terms around different time steps \cite{bucket_sim}. This latter concept was first developed with the idea of ``coalescing" terms into ``buckets" by Wecker et al. \cite{bucket_sim} and further explored by Poulin et al. \cite{coalescing_con_wiebe}. They showed that grouping terms of similar sizes together to be skipped during certain Trotter steps led to negligible increases in error and reduced gate counts by about a factor of 10.  Similar improvements are also seen in the randomized setting of~\cite{kivlichan2019phase}. In this work we extend on these ideas by placing a specific set of terms into a Trotter partition and the rest in a QDrift partition. This simple division can then be studied analytically and we are able to provide sufficient conditions on asymptotic improvements over completely Trotter or completely QDrift channels. Although we are not able to develop the idea of conditional samples in QDrift protocols, our 
procedure can be viewed as a specific subset of what a generic Markovian QDrift would look like. We briefly mention these generalizations in 
Section \ref*{sec:discussion}. 
 
Recent approaches have sought to use the advantages of randomized compilation as a subset of an overall simulation, such as the hybridized
scheme for interaction picture simulations \cite{hybridized_interaction_pic}. What separates these two works is that our approach offers a
more flexible approach for generic time-independent simulation problems whereas the hybridized schemes are specifically tailored to taking
advantage of the time dependence introduced by moving to an interaction picture. As such, the hybridized approach achieves asymptotic advantages
when the size of the interaction picture term dominates the overall Hamiltonian. This typically occurs in instances in which the size of an operator
is unbounded, which can occur in lattice field theory simulations or constrained systems. The way the hybridized scheme in 
\cite{hybridized_interaction_pic} works is via a ``vertical" stacking of simulation
channels, for example one channel to handle the Interaction Picture rotations and then other channels on top of this to simulate the time-dependence it generates on the remaining Hamiltonian terms. Our work instead remains in the Schrodinger time evolution picture and we 
perform a ``horizontal" stacking of simulation techniques. By horizontal we mean for a given simulation time we split the Hamiltonian up into 
(potentially) disjoint partitions and simulate each partition for the full simulation time but with different techniques, such as Trotter or QDrift.
These techniques allow us to achieve asymptotic improvements over either method for a loose set of assumptions.

There are two other simulation techniques that have been proposed recently that have a similar interpolation behavior between QDrift and Trotter channels. The first of these methods is the SparSto, or Stochastic Sparsification, technique by Ouyang, White, and Campbell \cite{sparsto}. The procedure \cite{sparsto} randomly sparsifies the Hamiltonian and performs a randomly ordered first-order Trotter formula on the sampled Hamiltonian. They construct these probabilities such that the expected Hamiltonian is equal to the Hamiltonian being simulated. They then fix the expected number of oracle queries of the form $e^{i H_i t'}$ and give diamond distance bounds on the resulting channel error. The claim for interpolation between Trotter and QDrift is that one can fix the expected number of gates to be 1 for each time step, in which case the sparsification mimics QDrift, whereas if no sparsification is performed then the channel is simply implementing Trotter. They show that this allows for one to have reduced simulation error up to an order of magnitude on numerically studied systems as compared to Trotter or QDrift. One downside to these techniques is that the number of gates applied is a random variable, so making gate cost comparisons is rather difficult especially considering that no tail bounds on high gate cost sampled channels are provided. In \cite{sparsto} they prefer to fix the expected gate cost and analyze the resulting diamond norm error. In contrast, our procedures directly implement both QDrift and Trotter channels and have a fixed, deterministic gate cost.

The second method of note with both QDrift and Trotter behavior is that of Jin and Li \cite{jin2021partially}. They develop an analysis of the variance of a unitary consisting of a first-order Trotter sequence followed by a QDrift channel. They focus on bounding the Mean Squared Error (MSE) of the resulting channel and use a simple partition of the Hamiltonian terms based on spectral norm. Their partitioning scheme places all terms below some cutoff into the first-order Trotter sequence and all terms above the cutoff into the QDrift channel. Their main results show an interpolation of the MSE between 0 when the partitioning matches a solely Trotter channel and matching upper bounds for QDrift when all terms are randomly sampled. This work goes beyond the results from Jin and Li by providing an analysis of the diamond distance between an ideal evolution and our implemented channel, which is more useful analytically than the MSE, as well as providing upper bounds on the number of gates needed in an implementation to meet this diamond distance. In addition our work remains independent of specific partitioning schemes as much as possible and instead places restrictions on which partitions achieve improvements. In the interest of practicality we do show methods for partitioning that can be useful in both the first-order and higher-order Trotter cases. Specifically for higher-order Trotter formulas we give a probabilistic partitioning scheme that is easily computable and matches gate cost upper bounds in the extreme limits as our probabilities saturate the QDrift and Trotter limits. 

The rest of the paper is organized as follows. We first provide a brief summary of the main results in Section \ref{sec:main}. After reviewing known results and notation in Section \ref{sec:prelim}, we explore methods for creating Composite channels using First-Order Trotter Formulas with QDrift in Section \ref{sec:first_order_trotter} as a warmup. This is broken down
into three parts in which we find the gate cost for an arbitrary partition, we then give a method for producing a good partitioning, and then we analyze conditions in which a Composite channel can beat either first-order Trotter or QDrift channels. In Section \ref{sec:higher_order_trotter} we then extend this framework to more general higher-order Trotter Formulas. This section mirrors the organization of the first-order Trotter section,
namely we find the cost of an arbitrary partition, we give a method for producing a partition efficiently, and then we analyze when one could see
improvements over the constituent channels. Finally, in Section \ref{sec:discussion} we discuss extensions to this model that allow a flexible interpolation between various types of product formulas that could be leveraged numerically. 

\section{Main Results} \label{sec:main}
In this section we summarize the gate cost performance of a higher-order Composite channel, the probabilistic partitioning scheme developed, and the conditions needed for a partitioning scheme to satisfy in order to expect asymptotic improvements over QDrift or Trotter. These results are motivated and proved throughout Section \ref{sec:higher_order_trotter}. We do not state the first-order Composite channel results here as they are more specific and do not achieve as strong asymptotic improvements as the higher-order channels.

Our first theorem, presented in Section \ref{sec:higher_order_complexity}, gives an upper bound on the number of queries to oracles of the form $e^{i H_i t'}$ to implement a Composite channel with desired error. We provide a conceptually simple packaging of this bound in terms of the number of oracle queries that would be needed to perform this same simulation if either Trotter or QDrift had been used alone.
\begin{restatable}[Gate Cost for Higher-Order Composite Channel]{thm}{assCost} \label{thm:higher_order_cost_fixed}
Given a time $t$, error bound $\epsilon$, partitioned Hamiltonian $H = A + B$, and let $\widetilde{\capU}^{(2k)}$ denote the higher-order Composite channel to approximate the exact unitary evolution $\capU(t)$. By using $r$ iterations of $\widetilde{\capU}^{(2k)}(t/r)$ we can satisfy the error requirement $\diamondnorm{\capU(t) - \widetilde{\capU}^{(2k)}(t/r)^{\circ r}} \leq \epsilon$ by using at most the following number of operator exponentials
\begin{align}
    &C_{comp}(A,B, t, \epsilon, 2k) \nonumber \\
    &\leq \Upsilon (\Upsilon L_A + N_B) \ceil{\frac{(\Upsilon t)^{1 + 1/2k} 4^{1/2k}}{\epsilon^{1/2k}} \parens{\frac{\Upsilon \alpha_{comm}(A, 2k) + \alpha_{comm}(\set{A,B}, 2k)}{2k+1}}^{1/2k} + \frac{4 \Upsilon \lambda_B^2 t^2}{N_B \epsilon}}.
\end{align}
By making the definition $q_B \coloneqq\frac{\alpha_{comm}(B,2k)}{\alpha_{comm}(H, 2k)}$ and utilizing the upper bounds from Theorems \ref{thm:trotter_cost} and \ref{thm:QDrift}, where $C_{Trott}$ and $C_{QD}$ below are  upper bounds on the number of operator exponentials required in the Trotter and QDrift channels, we can write the cost upper bound as
\begin{align}
    C_{comp} \leq \Upsilon (\Upsilon L_A + N_B) \ceil{C_{Trott}(H, t, \epsilon, 2k) \frac{(1-q_B)^{1/2k}}{\Upsilon^{1- 1/2k}L} + C_{QD}(H, t, \epsilon) \frac{\Upsilon}{N_B} \frac{\lambda_B^2}{\lambda^2}}.
\end{align}
\end{restatable}

The next result, presented in Section \ref{sec:probabilistic_partitioning}, gives an easily computable probabilistic partitioning scheme for an arbitrary Hamiltonian. It is based around a probability for each term to end up in either the QDrift partition or the Trotter partition that can be viewed as an importance sampling routine on the inverse spectral norms $\frac{1}{h_i}$. This distribution was motivated by upper bounding the expected QDrift error, which is expressed in terms of $\lambda_B$, in terms of the Trotter error as intuitively Trotter formulas of higher-orders are more accurate than QDrift channels at smaller times. One feature of note from this lemma is that it introduces a \emph{lower} bound on the number of QDrift samples which is required for the probabilities to remain nonnegative.
\begin{restatable}[Probabilistic Partitioning Scheme]{lemmer}{assProb} \label{lem:prob_lemma}
For a composite simulation of $H$ for time $t$ and error $\epsilon$, let $p_i$ denote the probability of placing term $h_i H_i$ into the Trotter partition of Composite channel. we have that choosing 
\begin{equation}
    1-p_i = \min\set{\frac{\lambda}{h_i L} \parens{ \sqrt{ N_B \parens{\frac{\epsilon}{\lambda t}}^{1-1/2k} \parens{\frac{2k+\Upsilon}{2k+1}}^{1/2k} \frac{\Upsilon^{1/2k}}{2^{1-1/k}} } - 1},1} \eqqcolon \min \set{\frac{\chi}{h_i},1}, \label{eq:prob_def}
\end{equation}
along with choosing the number of QDrift samples $N_B$ to satisfy
\begin{equation}
    N_B \geq \parens{\frac{\lambda t}{\epsilon}}^{1 - 1/2k} \parens{\frac{2k + 1}{2k + \Upsilon}}^{1/2k} \frac{2^{1-1/k}}{\Upsilon^{1/2k}}
\end{equation}
guarantees the following:
\begin{enumerate}
    \item $p_i \in [0,1]$,
    \item the expectation value of the coefficients in the QDrift partition satisfies 
    $$\frac{ \expect{\lambda_B} }{\lambda} \leq \frac{1}{2} \sqrt{\parens{\frac{4k + 2 \Upsilon}{2k + 1}}^{1/2k} (2 \Upsilon)^{1+1/2k}} \sqrt{N_B \parens{\frac{\epsilon}{\lambda t}}^{1-1/2k}} .$$ This bound follows from a rigorous interpretation of making the QDrift and Trotter errors approximately equivalent.
\end{enumerate}
\end{restatable}

Our final result, proved in Section \ref{sec:higher_order_improvements}, of significant importance are the conditions on the parameters of a given partition to provide asymptotic cost improvements of a Composite channel over it's constituent channels, Trotter or QDrift. This theorem is relevant when one is considering a family of Hamiltonians and has a provided partitioning scheme for each. We then study how these partitions can give rise to asymptotically superior simulation techniques over Trotter and QDrift as the number of terms in the Hamiltonian is taken to infinity. Although the situation of an infinite family of Hamiltonians with a provided partitioning scheme may not arise in practice, analyzing the performance of these families provides intuition as for when a Composite channel should be able to yield significant savings.
\begin{restatable}[Conditions for Composite Channel Improvements]{thm}{assImprovements} \label{thm:higher_order_improvements_general}
    Assume a Hamiltonian $H$ along with a partitioning scheme to generate $A$ and $B$ that varies with $L$. For a simulation time $t$ with a desired diamond 
    distance error at most $\epsilon$, let $\beta > 0$ be a number such that $C_{QD} = C_{Trott}^\beta$. There exist asymptotic regimes for the parameters $L_A, \lambda_B,$ and $N_B$ such that $$C_{Comp} \in o(\min \set{C_{Trott}, C_{QD}}),$$ outlined below for the cases when $C_{QD} \geq C_{Trott}$  $(\beta \geq 1)$ and $C_{QD} < C_{Trott}$ $(0 < \beta < 1)$. 
    
    For the case when $\beta > 1$, indicating Trotter uses fewer queries than QDrift, if the parameters $\lambda_B$, $L_A$, and number of QDrift samples $N_B$ satisfy the following
    \begin{enumerate}
        \item $L_A (1 - q_B)^{1/2k} \in o(L)$ where $q_B = \frac{\alpha_{comm}(B, 2k)}{\alpha_{comm}(H, 2k)}$,
        \item $\lambda_B \in o \parens{\lambda^{1/\beta} \parens{\frac{\sqrt{\epsilon}}{t}}^{1-1/\beta}}$,
        \item $N_B \in \Omega(L_A)$ and $N_B \in o\parens{\frac{L}{(1-q_B)^{1/2k}}}$,
    \end{enumerate}
    then we have that $C_{comp} \in o(C_{Trott})$.
    
    If instead $0 < \beta < 1$, indicating $C_{QD} < C_{Trott}$, and the parameters
    $\lambda_B$, $L_A$ and $N_B$ satisfy the following
    \begin{enumerate}
        \item the total number of terms in the Trotter partition satisfies
        \begin{equation}
            L_A \in o\parens{ L^{1/\beta} \parens{\frac{\epsilon^{1-1/\beta}}{t^{(2k+1)(1-1/\beta)}} \frac{\alpha_{comm}^{1/\beta}(H)}{\alpha_{comm}(A) + \alpha_{comm}(\set{A,B})} }^{1/2k}},
        \end{equation}
        \item $\lambda_B \in o(\lambda)$,
        \item and $N_B \in \Theta(L_A)$,
    \end{enumerate}
    then we have that $C_{comp} = o(C_{QD})$. Note that for $\beta = 1$ exactly, the conditions on $\lambda_B$ and $L_A$ are the same for both cases:
    $L_A \in o(L)$ and $\lambda_B \in o(\lambda)$. The conditions for $N_B$ are satisfied by $N_B \in \Theta(L_A)$.
    If these conditions are satisfied for all ranges of $\beta$ then $C_{comp} \in o \parens{\min \set{C_{QD}, C_{Trott}}}$. 
\end{restatable}

\subsection*{Conditions for Advantage for Simulations with Composite Channels} 
To help give an idea of when a Composite channel would be most effective we briefly and informally discuss the intuition provided throughout the paper of when one should expect to see cost savings from a Composite channel. The most straightforward tool to build intuition is the situation in which $t$ and $\epsilon$ are such that $C_{Trott}(H, t, \epsilon) = C_{QD}(H, t, \epsilon)$. We will refer to this particular ratio of $t$ and $\epsilon$ as the cost crossover time. In this setting, all one has to do is find a partitioning scheme such that the number of Trotter terms is much smaller than the total number of terms ($L_A \in o(L)$) and that the spectral norm of the remaining terms is negligible compared to the overall sum ($\lambda_B \in o(\lambda)$). Note that this last expression can be rewritten: $\lambda - \lambda_A \in o(\lambda) \implies 1 - \lambda_A / \lambda \in o(1)$. Combining these pieces of information tells us that the Composite framework should provide the best improvements whenever a vanishingly small number of terms contain almost all of the ``spectral weight" of the Hamiltonian and have negligible commutator structure.

When considering the task of how to partition a given Hamiltonian we can unfortunately not offer more insight beyond the intuition provided above. This is likely to be a very difficult problem that will have to take advantage of domain specific knowledge in regards to the provided Hamiltonian. A possible starting point to constructing partitions could be gleaned from the effectiveness of our provided probabilistic partitioning scheme when applied to an exponentially decaying Hamiltonian. This probabilistic scheme depends solely on the spectral norm of each term, so a useful starting point for constructing deterministic partitions could be to pick a cutoff weight in which stronger terms are assigned to Trotter and lighter ones to QDrift. If a cutoff can be found that is small compared to the total norm $\lambda$ and only has a small percentage of terms (roughly $\log_2 (L) / L$ would align with Theorem \ref{thm:exponential_decay}), then this should be enough to see significant improvements in cost for simulation times near the cost crossover time. We note that many ``toy" chemistry and material science models, such as Jellium \cite{babbush2018low} for interacting electrons and Hydrogen chains \cite{whitfield2011simulation} for molecules, exhibit these kinds of strong decays in spectral norms.

Another useful consideration is the question of when a given simulation is likely to \emph{not} see significant savings from a composite approach. When considering Hamiltonian norms and commutator structure the worst case scenario is one in which each term has equal spectral norm and there is a cyclic commutator behavior (e.g. angular momentum $[J_i, J_j] = \epsilon_{ijk} J_k)$ that is the same magnitude at any order. We also note that the ability to find useful partitions depends heavily on the simulation time to error ratio $t / \epsilon$. As this ratio tends to either 0 or $\infty$ the ability to find an economical partitioning vanishes. For example if one needs a very accurate simulation ($\epsilon \to 0$) then any terms that are put in a QDrift partition will require too many samples to meet the lower error budget compared to just putting the term in a higher-order Trotter formula. At the other extreme in which one has a higher error or very short time $t$, then any error savings by putting terms into a Trotter partition are likely to waste gates when sampling these terms with QDrift would suffice.


\section{Preliminaries} \label{sec:prelim}
In this section we will first introduce the necessary notation we will use and then state known results about Trotter-Suzuki formulas and QDrift
channels. We work exclusively with time-independent Hamiltonians $H$ in a $2^n$ dimensional Hilbert space $\hilbSpace$. We also assume that $H$
consists of $L$ terms $H = \sum_{i = 1}^L h_i H_i$ where $h_i$ represents the spectral norm of the term, $H_i$ is a Hermitian operator on $\hilbSpace$, and $\norm{H_i} = 1$. Note without loss of generality we can always assume $h_i \geq 0$, as we can always absorb the phase into the 
operator $H_i$ itself. We use $\norm{M}$ to refer to the spectral norm, or the magnitude of the largest singular value of $M$. We use $\lambda$ to refer to the sum of $h_i$, namely $\lambda = \sum_i h_i$. We will also use subscripts on lambda, such as $\lambda_A$ to refer to sums of subsets of the terms of $H$. For example, if $H = 1 H_1 + 2 H_2 + 3 H_3$ and $G = 1 H_1 + 2 H_2$, then $\lambda = 6$ and $\lambda_G = 3$. 

We use $U(t)$ to refer to the unitary operator $e^{iHt}$ and $\capU(t)$ to refer to the channel $U(t) \rho U(t)^\dagger$. We will be particularly
concerned with simulations of subsets of the terms of $H$, which we denote as follows. We typically work with a partition of $H$ into two matrices
$H = A + B$, and we let $A = \sum_i a_i A_i$ and $B = \sum_j b_j B_j$, where we have simply relabeled the relevant $h_i$ and $H_i$ into $a$'s, $b$'s, $A$'s, and $B$'s. This allows us to define the exact unitary time evolution operators $U_A(t) = e^{i A t}$ and channels $\capUA(t) = U_A(t) \rho U_A(t)^\dagger$, similarly defined for $B$. As we will be working with approximations to these channels, any operator or channel with a tilde represents
an ``implemented" channel, for example a first-order Trotter operator for $A$ would look like $\widetilde{U_A}(t) = e^{i a_1 A_1 t} \ldots e^{i a_L A_L t}$. We avoid using $\mathcal{E}$ to represent an approximation or product formula as $\mathcal{E}$ will be used for error channels.

Although much of the literature for Trotter-Suzuki formulas is written in terms of unitary operators $U = e^{i H t}$ acting on state vectors $\ket{\psi}$ for our purposes it will prove most natural to consider a product formula as a channel $\mathcal{U} = e^{iHt} \rho e^{-iHt}$ acting on a density matrix $\rho$. After reviewing known results on unitary constructions of Trotter-Suzuki formulas we give a straightforward extension
of these bounds to channels. 

\subsection{Trotter-Suzuki Formulas}
\begin{definition}[Trotter-Suzuki Decomposition \cite{suzuki}]\label{def:TS}
Given a Hamiltonian $H$, let $\trotterU{1}(\rho; t)$ denote the first-order Trotter-Suzuki time evolution operator, which is defined as 
\begin{equation}
    \trotterU{1}(t) \coloneqq e^{i h_L H_L t} \ldots e^{i h_1 H_1 t}  = \prod_{i = 1}^{L} e^{i h_i H_i t}.
\end{equation}
Note that the ordering of the factors in the notation $\prod_{i = 1}^L$ is defined to start from the rightmost operator and end at the leftmost. Following this we can define the second-order Trotter-Suzuki time evolution operator as
\begin{align}
    \trotterU{2}(\rho; t) &\coloneqq e^{i h_1 H_1 \frac{t}{2}} \ldots e^{i h_L H_L  \frac{t}{2}} e^{i h_L H_L  \frac{t}{2}} \ldots e^{i h_1 H_1 \frac{t}{2}} \\
    &= \prod_{i = L}^{1} e^{i h_i H_i \frac{t}{2}} \prod_{j = 1}^{L} e^{i h_j H_j \frac{t}{2}}.
\end{align}
This formula serves as the base case for the higher-order formulas which can be written as 
\begin{equation}
    \trotterU{2k}(t) \coloneqq \trotterU{2k-2}(u_k t)^2 \cdot \trotterU{2k-2}((1-4 u_k)t) \cdot \trotterU{2k-2}(u_k t)^2,
\end{equation}
where $u_k \coloneqq 1 / \parens{4-4^{1/(2k - 1)}}$. In addition we define $\Upsilon \coloneqq 2 \cdot 5^{k-1}$ as the number of "stages" in the higher-order product formula. We can now introduce the time evolution channels as
\begin{equation}
    \curlyP{2k}(\rho; t)\coloneqq \trotterU{2k}(t) \rho \trotterU{2k}^\dagger(t),
\end{equation}
where for consistency we use the calligraphic $\curlyP{2k}$ to represent the applied channels.
\end{definition}

Before we introduce the cost and error scaling for Trotter-Suzuki formulas we will make use of the following notation that captures information about the commutator structure of a Hamiltonian or subset of terms from a Hamiltonian. First used in Childs et. al \cite{childs2021theory} the sum of norms of commutators $\alpha_{comm}$ (represented as $\widetilde{\alpha}_{comm}$ in \cite{childs2021theory}) is given by
\begin{equation}
    \alpha_{comm}(H, 2k) \coloneqq \sum_{\gamma_i \in \set{1,\ldots, L}} \parens{\prod h_{\gamma_i}} \norm{[H_{\gamma_{2k+1}}, [H_{\gamma_{2k}},\ldots[H_{\gamma_2}, H_{\gamma_1}]\ldots]}_{\infty}. \label{def:alpha_comm}
\end{equation}
Another variation we will make is the restriction of $\alpha_{comm}$ to subsets of $H$, for example if we can form two subsets $A, B$ of $H$ such that $H = A+ B$ then we can write the following
\begin{equation}
    \alpha_{comm}(A, 2k) = \sum_{\gamma_i \in \set{1,\ldots L}} \parens{\prod a_{\gamma_i}} \norm{[A_{\gamma_{2k+1}}, [A_{\gamma_{2k}},\ldots[A_{\gamma_2}, A_{\gamma_1}]\ldots]}_{\infty}.
\end{equation}
We can then define the commutator structure between the two subsets as all nested commutators that contain \emph{at least} one term from both $A$ and $B$. This then allows for the expression
\begin{equation}
    \alpha_{comm}(\set{A,B}, 2k) = \alpha_{comm}(H, 2k) - \alpha_{comm}(A, 2k) - \alpha_{comm}(B, 2k),
\end{equation}
as any nested commutator with only terms consisting of $A$  matrices is contained in $\alpha_{comm}(A, 2k)$ and similarly for $B$.

Now we can state a summary of the performance of Trotter-Suzuki formulas as proved in \cite{childs2021theory}.
\begin{theorem}[Trotter-Suzuki Formulas] \label{thm:trotter_cost}
    Given a Hamiltonian $H$, time $t$, and error bound $\epsilon$, a $2k\ts{th}$-order Trotter-Suzuki channel as defined in \refeq{def:TS} satisfies $\diamondnorm{\capU(t) - \curlyP{2k}(t)} < \epsilon$ and uses the following number of gates of the form $e^{i H_i t'}$
    \begin{equation}
        C_{Trott}(H, t, \epsilon, 2k) = \Upsilon L r \leq \Upsilon L \ceil{\frac{(\Upsilon t)^{1+1/2k}}{\epsilon^{1/2k}} \parens{\frac{4 \alpha_{comm}(H, 2k)}{2k+1}^{1/2k}}}.
    \end{equation}
    Similarly, a first-order Trotter-Suzuki formula has the following cost
    \begin{equation}
        C_{Trott}(H, t, \epsilon, 1) = L r \leq L \ceil{\frac{t^2}{2 \epsilon} \sum_{i, j} h_i h_j \norm{[H_i, H_j]}_{\infty}}
    \end{equation}
\end{theorem}
\begin{proof}
    We first upper bound the diamond distance between our implementation and the ideal time evolution channels as follows.
    \begin{align}
        &\diamondnorm{\capU(t) - \curlyP{2k}(t)} \coloneqq \norm{\parens{\capU(t) - \curlyP{2k}(t)}\otimes \openone}_1 \label{eq:diamond_to_spectral_start} \\
        =& \max_{\rho : \norm{\rho}_1 \leq 1} \norm{\parens{e^{iHt}\otimes \openone} \rho \parens{e^{-i H t}\otimes \openone} - \parens{\trotterU{2k}(t)\otimes \openone} \rho \parens{\trotterU{2k}^\dagger(t)\otimes \openone} }_1 \\
        \leq& \max_{\rho : \norm{\rho}_1 \leq 1} \norm{\parens{e^{iHt}\otimes \openone} \rho \parens{e^{-i H t}\otimes \openone} - \parens{e^{i H t} \otimes \openone} \rho \parens{\trotterU{2k}^\dagger(t) \otimes \openone} }_1 \\
        ~& +\max_{\rho : \norm{\rho}_1 \leq 1} \norm{ \parens{e^{i H t} \otimes \openone} \rho \parens{\trotterU{2k}^\dagger(t) \otimes \openone}  - \parens{\trotterU{2k}(t)\otimes \openone} \rho \parens{\trotterU{2k}^\dagger(t)\otimes \openone} }_1 \nonumber \\
        =& \max_{\rho : \norm{\rho}_1 \leq 1} \norm{\rho \parens{e^{-i H t} - \trotterU{2k}^\dagger(t)}\otimes \openone}_1 + \max_{\rho : \norm{\rho}_1 \leq 1} \norm{\parens{e^{i H t} - \trotterU{2k}(t)}\otimes \openone \rho}_1 \\
        \leq & 2 \norm{e^{i H t} - \trotterU{2k}(t)}_\infty \max_{\rho : \norm{\rho}_1 \leq 1} \norm{\rho}_1 \\
        = & 2 \norm{e^{i H t} - \trotterU{2k}(t)}_\infty. \label{eq:TS_intermediate_1}
    \end{align}
    We can then make use of Eq. (189) and Theorem 10 from \cite{childs2021theory}, which provides the following bound
    \begin{equation}
        \norm{e^{iHt/r} - \trotterU{2k}(t/r)}_{\infty} \leq \frac{2 \alpha_{comm}(H, 2k)}{2k+1} \parens{\frac{\Upsilon t}{r}}^{2k+1}. \label{eq:TS_intermediate_2}
    \end{equation}
    We note that this equation differs from Eq. (189) in \cite{childs2021theory} due to the different $\alpha_{comm}$ used. The denominator of $(2k+1)!$ is replaced by $2k+1$ due to a factor of $(2k)!$ from upper bounds on the $\alpha_{comm}$ used in \cite{childs2021theory}. Note that this also leads to the extra factors of $\Upsilon^{2k}$, as opposed to just $\Upsilon$ in Eq. (189) in \cite{childs2021theory}.
    
    For the first-order formula we will use the following upper bound which follows from an application of the triangle inequality to Eq. (143) from \cite{childs2021theory}
    \begin{equation}
        \norm{e^{iHt/r} - \trotterU{1}(t/r)}_{\infty} \leq \frac{t^2}{2r^2} \sum_{i, j} h_i h_j \norm{[H_i, H_j]}_{\infty} 
    \end{equation}
    Combining Eqs. \eqref{eq:TS_intermediate_1} and \eqref{eq:TS_intermediate_2}, along with the inequality $\diamondnorm{X^{\circ r} - Y^{\circ r}} \leq r \diamondnorm{X - Y}$, yields
    \begin{align}
        \diamondnorm{\capU(t) - \curlyP{2k}(t/r)^{\circ r}} &\leq r \diamondnorm{\capU(t/r) - \curlyP{2k}(t/r)} \\
        &\leq \frac{4 r \alpha_{comm}(H, 2k)}{2k+1} \parens{\frac{\Upsilon t}{r}}^{2k+1}. \label{eq:trotter_diamond_error}
    \end{align}
    We then can require the inequality in Eq. \eqref{eq:trotter_diamond_error} to be less than $\epsilon$ and solve for $r$, yielding 
    \begin{equation}
        r > \frac{\parens{\Upsilon t}^{1+1/2k}}{\epsilon^{1/2k}} \parens{\frac{4 \alpha_{comm}(H, 2k)}{2k+1}}^{1/2k}. \label{eq:TS_intermediate_3}
    \end{equation}
    By taking the ceiling of the RHS of \eqref{eq:TS_intermediate_3} and plugging the result into $C_{Trott}(H, t, \epsilon) = \Upsilon L r$ yields the expression in the statement. Similar results hold for the first-order case. 
\end{proof}

\subsection{QDrift}
We now shift our attention to the other main product formula we will make use of, that is QDrift. Introduced by Campbell in \cite{qdrift}, the main premise of QDrift is that one randomly picks a term $H_i$ from the overall set of terms according to the ratio of spectral norms $\nicefrac{h_i}{\lambda}$ and then apply the exponential gate $e^{i H_i \tau}$, for some $\tau \propto t$. This is summarized in the following definition.
\begin{definition}[QDrift Channel] \label{def:qdrift_channel}
    Let $p_i = \frac{h_i}{\lambda}$, where $\lambda = \sum_i h_i$, represent a probability distribution over terms in a Hamiltonian $H = \sum_i h_i H_i$. We define the QDrift channel for a single sample from this distribution as
    \begin{equation}
        \qdchan(t) : \rho \mapsto \sum_i p_i e^{i H_i \lambda t} \rho e^{- i H_i \lambda t}.
    \end{equation}
\end{definition}

Below we restate the main results from \cite{qdrift}, in which multiple independent samples of the above channel are studied, with only minor modifications to the allowable range of $\epsilon$.
\begin{theorem}[QDrift] \label{thm:QDrift}
    Given a Hamiltonian $H$, time $t$, and error bound $\epsilon$ one can approximate the ideal unitary dynamics of $\capU(t)$ by taking $N$ i.i.d samples of the QDrift channel from Definition \ref{def:qdrift_channel}. To meet the error bound $\epsilon$, namely $\diamondnorm{\capU(t) - \qdchan(\nicefrac{t}{N})^{\circ N}} < \epsilon$, it suffices to choose $N = \nicefrac{4 t^2 \lambda^2}{\epsilon}$ if we restrict allowed values of $\epsilon$ to within the range $(0, \lambda t \ln (2)/2)$. This gives the cost of the channel, or the number of operator exponentials of the form $e^{i H_i t'}$ that the channel requires, as
    \begin{equation}
        C_{QD}(H, t, \epsilon) \leq \frac{4 \lambda^2 t^2}{\epsilon}.
    \end{equation}
\end{theorem}

Our minor modification follows from the proof of the following expression
\begin{equation}
    \diamondnorm{\capU(t) - \qdchan(t/N)^{\circ N}} \leq \frac{2 \lambda^2 t^2}{N} e^{2 \lambda t / N}, \label{eq:qdrift_diamond_distance}
\end{equation}
which was given in \cite{qdrift}. We upper bound the coefficient $e^{2 \lambda t/N} \leq 2$ by using $N = 4 \lambda^2 t^2 /\epsilon$ and restricting $\epsilon \in (0, \lambda t \ln(2)/2)$.


\section{First-Order Trotter with QDrift}\label{sec:first_order_trotter}
The most straightforward Composite channel to analyze is combining a first-order Trotter formula with a QDrift channel. We proceed in four steps. First, we assume a partitioning $H = A + B$ which allows us to determine the diamond distance error scaling of the Composite channel. Next, we formulate an upper bound on the number of exponential gates of the form $e^{i H_i t}$ needed to achieve this error. Following this, we use the derived cost function to determine a useful partitioning scheme for determining whether a term in a given Hamiltonian $H$ should end up in the Trotter channel or the QDrift channel. Finally, we give an instance in which a Composite channel can offer asymptotic improvements over either a purely Trotter or QDrift
channel.

\subsection{Query Complexity}
To analyze the error of our Composite channel we need to first reduce the overall time evolution channel $\rho \mapsto e^{-iHt} \rho e^{+iHt}$ into the simpler pieces that we can analyze with our Trotter and QDrift results. Assuming a partitioning $H = A + B$, where $A$ consists of terms that we would like to simulate with Trotter and $B$ has the terms we would like to sample from with QDrift. We now introduce the ``outer-loop" error $E_{\set{A,B}}$ induced by this partitioning, which is as follows
\begin{equation}
    e^{-iHt} \rho e^{+iHt} = e^{-iBt}e^{-iAt} \rho e^{+iAt} e^{+iBt} + E_{\set{A,B}}(t).
\end{equation}
We use the phrase ``outer-loop" as this decomposition is done before any simulation channels are implemented. 

The rest of the error analysis is captured in the following lemma.
\begin{theorem}[First-Order Composite Channel] \label{thm:first_order_composite}
Given a time $t$, error bound $\epsilon$, and a partitioned Hamiltonian $H = A + B$ one can construct a first-order composite simulation channel $\widetilde{\capU}(t)$ that approximates the ideal channel $\capU(t)$ within a diamond distance $\epsilon$ as follows. Let $\widetilde{\capU}(t) = \widetilde{\capUB}(t) \circ \widetilde{\capUA}(t)$ represent the composition of a Trotter-Suzuki channel $\widetilde{\capUA}$ to simulate $A$ and a QDrift channel $\widetilde{\capUB}$ for $B$. By repeating $\widetilde{\capU}(t/r)$ for $r$ iterations, the diamond distance bound $\diamondnorm{\capU(t) - \widetilde{\capU}(t/r)^{\circ r}} < \epsilon$ can be acheived by using no more than 
\begin{align}
    C_{Comp}(A, B, t, \epsilon) &= (L_A + N_B) r \\
    &= (L_A + N_B) \ceil{\frac{t^2}{\epsilon} \parens{\sum_{i,j} a_i a_j \norm{[A_i, A_j]}_{\infty} +  \sum_{i,j} a_i b_j\norm{ [A_i, B_j]}_{\infty} + \frac{4 \lambda_B^2}{N_B}}} \label{eq:first_order_comp_cost}
\end{align} 
gates of the form $e^{i H_i t'}$. 

\end{theorem}
\begin{proof}
We will need to make use of the following minor result 
\begin{equation}
    \diamondnorm{X^{\circ r} - Y^{\circ r}} \leq r \diamondnorm{X - Y}, \label{eq:err_to_err_per_iter}
\end{equation}
where $X$ and $Y$ are channels. This follows straightforwardly from subadditivity of the diamond norm with respect to composition of channels. Now starting with the outer-loop decomposition mentioned above we can reduce the overall channel distance to a per-iteration distance as follows
\begin{align}
    \diamondnorm{\capU(t) - \widetilde{\capU}\parens{\nicefrac{t}{r}}^{\circ r}} &\leq \diamondnorm{\parens{\capUB(\nicefrac{t}{r}) \circ \capUA(\nicefrac{t}{r}) + E_{\set{A,B}}(\nicefrac{t}{r})}^{\circ r} - \parens{\widetilde{\capUA}(\nicefrac{t}{r}) \circ \widetilde{\capUB}(\nicefrac{t}{r})}^{\circ r}} \\
    &\leq r \parens{ \diamondnorm{\capUA(\nicefrac{t}{r}) - \widetilde{\capUA}(\nicefrac{t}{r})} + \diamondnorm{\capUB(\nicefrac{t}{r}) - \widetilde{\capUB}(\nicefrac{t}{r})} + \diamondnorm{E_{\set{A,B}}(\nicefrac{t}{r})} }. \label{eq:dmd_bnd_intermediate_1}
\end{align}
This is now in a form where we can use the results from Section \ref{sec:prelim} for Trotter formulas and QDrift channels. We use the QDrift results from Theorem \ref{thm:QDrift} that $\diamondnorm{\capUB(t/r) - \widetilde{\capUB}(t/r)} \leq \frac{4 \lambda_B^2 t^2}{N_B r^2}$ and Eqs \eqref{eq:diamond_to_spectral_start} - \eqref{eq:TS_intermediate_1} from Theorem \ref{thm:trotter_cost}  to reduce $\diamondnorm{\capUA(t/r) - \widetilde{\capUA}(t/r)} \leq 2 \norm{e^{iAt/r} - \trotterU{1}(t/r)}_{\infty}$. The last term we need to bound is the outer-loop error
\begin{align}
    \diamondnorm{E_{\set{A, B}}(\nicefrac{t}{r})} &= \diamondnorm{\capU(\nicefrac{t}{r}) - \capUB(\nicefrac{t}{r}) \circ \capUA(\nicefrac{t}{r})} \\
    &\leq 2 \norm{e^{i H (t/r)} - e^{i B (t/r)} e^{i A (t/r)}}_{\infty} \\
    &\leq \frac{t^2}{r^2} \sum_{i,j} a_i b_j \norm{[A_i, B_j]}_{\infty}
\end{align}

Plugging in Theorems \ref{thm:QDrift} and \ref{thm:trotter_cost} into Eq. \ref{eq:dmd_bnd_intermediate_1} yields
\begin{align}
    \frac{1}{r}\diamondnorm{\capU(t) - \widetilde{\capU}(\nicefrac{t}{r})^r} &\leq \parens{\frac{t}{r}}^2  \parens{\sum_{i,j} a_i a_j \norm{[A_i, A_j]}_{\infty} + \sum_{i, j} a_i b_j \norm{[A_i, B_j]}_{\infty}} +  \frac{4 \lambda_B^2 t^2 }{N_B r^2} \leq \frac{\epsilon}{r}, \label{eq:first_order_intermediate_1}
\end{align}
where $N_B$ represents the number of samples used by QDrift to simulate $e^{iBt}$. It is straightforward to solve for $r$ that satisfies the inequality in Eq \eqref{eq:first_order_intermediate_1} to plug into the expression $C_{Comp}(A, B, t, \epsilon) = (L_A + N_B) r$ which yields the theorem statement. 
\end{proof}


We now use a relaxation of the first-order Composite channel cost from Eq. \eqref{eq:first_order_comp_cost} in which we allow for non-integer values, which is given as 
\begin{equation}
\widetilde{C}_{Comp}(A, B, t, \epsilon) \coloneqq (L_A + N_B) \frac{t^2}{\epsilon} \parens{\sum_{i,j} a_i a_j \norm{[A_i, A_j]}_{\infty} + \sum_{i,j} a_i b_j \norm{[A_i, B_j]}_{\infty} + \frac{4 \lambda_B^2}{N_B}}.
\end{equation}
One unspecified quantity in the above expression is $N_B$ which is specifically left as a user-defined parameter. This means we can optimize the non-integer cost $\widetilde{C}_{Comp}$ with respect to $N_B$, which is done in the following lemma.
\begin{lemma}\label{lem:first_order_opt_nb}
Let $\widetilde{C}_{Comp}(t, \epsilon, A, B)$ denote the non-integer cost of a first-order Composite channel approximation to $\capU(t)$. Then $\widetilde{C}_{Comp}$ can be optimized with respect to $N_B$ when
\begin{equation}
    N_B = \sqrt{\frac{4 \lambda_B^2 L_A}{\parens{\sum_{i,j} a_i a_j \norm{[A_i, A_j]}_{\infty} + \sum_{i, j} a_i b_j \norm{[A_i, B_j]}_{\infty}}}},
\end{equation}
note that this expression is only defined if $\norm{[A_i, A_j]}_{\infty} > 0$ or $\norm{[A_i, B_j]}_{\infty} > 0$ for at least one $A_i$ or $B_j$.
\begin{proof}
The result follows from basic calculus with the additional assumption that we will treat the above cost upper bound as exact
\begin{align}
    \frac{\partial \widetilde{C}_{comp}}{\partial N_B} = \frac{t^2}{\epsilon} \brackets{\sum_{i,j} a_i a_j \norm{[A_i, A_j]} + \sum_{i, j} a_i b_j \norm{[A_i, B_j]}  - \frac{4 \lambda_B^2 L_A}{N_B^2} }.
\end{align}
Setting the above equal to zero and solving for $N_B$ yields the stated value. The second derivative can be shown as 
\begin{equation}
    \frac{\partial^2 \widetilde{C}_{comp}}{\partial N_B^2} = \frac{4 t^2 \lambda_B^2 L_A}{\epsilon N_B^3} \geq 0,
\end{equation}
which indicates the optima found is the minimal cost with respect to $N_B$.
\end{proof}
\end{lemma}

\subsection{Hamiltonian Partitioning} \label{sec:first_order_partitioning}
Now that we have upper bounded the number of operator exponentials needed for a Composite channel to satisfy $\diamondnorm{\capU(t) - \widetilde{\capU}(t)} < \epsilon$ with a predetermined partition we move on to the question of how to decide a partition. There are many different ways one could determine a partitioning, for example by using a greedy algorithm or a spectral norm based decider, and here we propose a new method that is based on our derived cost function. Our method allows one to take into account information about the commutation structure between terms and spectral norm information to compute a cost function gradient that can be minimized in a gradient descent approach. We also show an analytic minima of this gradient that allows for a greedy approach. 

The first step we have to take is to determine how to parametrize our cost function $\widetilde{C}_{Comp}$. We introduce new parameters $w_i$ which represent a weighting of each term $H_i$ between the Trotter and QDrift partitions. Starting with our Hamiltonian $H = \sum_i h_i H_i$ we rewrite each term as a parametrized sum, $h_i H_i \mapsto w_i h_i H_i + (1-w_i) h_i H_i$. Then we place all terms $w_i h_i H_i$ in the Trotter partition $A = \sum_i w_i h_i H_i$ and all the terms $(1-w_i) h_i H_i$ into the QDrift partition $B = \sum_i (1-w_i) h_i H_i$. Now instead of determining the discrete placement of each term into $A$ or $B$ we only need to determine an appropriate weighting $w_i$ of each term between the two partitions. As we would like the coefficients to remain positive after this remapping we require $w_i \in [0, 1]$. We will first work out the gradient of the cost with respect to 
each weight and then discuss its behavior. 

First we write the non-integer cost function of the weighted partitioned Composite channel as
\begin{equation}
    \widetilde{C}_{comp} = (L_A + N_B) \frac{t^2}{\epsilon} \parens{\sum_{i,j} w_i w_j h_i h_j \norm{[H_i, H_j]} + \sum_{i, j} w_i (1 - w_j) h_i h_j \norm{[H_i, H_j]} + \frac{4 \parens{\sum_i (1-w_i) h_i}^2 }{N_B} }, \label{eq:first_order_cost_partition}
\end{equation}
note that we are leaving $N_B$ as a user-defined integer and not the optimized value as found before. Now we can easily take the derivative of Eq. 
\eqref{eq:first_order_cost_partition} with respect to the $m\ts{th}$ weight $w_m$, which is
\begin{equation}
    \frac{\partial \widetilde{C}_{comp}}{\partial w_m} = (L_A + N_B) \frac{t^2}{\epsilon} \parens{ h_m \sum_j h_j \norm{[H_j, H_m]}_\infty -  \frac{8 h_m \sum_i (1-w_i) h_i}{N_B} } \label{eq:first_order_weighting_1}.
\end{equation}
This is enough information to perform a gradient descent to find a optima from an initial partitioning. However, it is relatively easy to
find the exact optima for a single Hamiltonian term with respect to the other weightings. We can set Eq. \eqref{eq:first_order_weighting_1} equal to 0 and solve for $w_m$ which yields
\begin{equation}
    \frac{\partial \widetilde{C}_{comp}}{\partial w_m} =0 \implies w_m = 1 - \sum_{i \neq m} \frac{h_i}{h_m} \parens{\frac{\norm{[H_i, H_m]}_{\infty}}{8} - (1-w_i)}. \label{eq:opt_first_order_weights}
\end{equation}

There are a few pieces of intuition we can gather from these expressions. First, if a term $H_m$ commutes with every other term in the Hamiltonian then $[H_i, H_m] = 0$ and $w_m = 1 + \sum_{i \neq m} \nicefrac{h_i}{h_m} (1-w_i)$, which is always greater than 1. Since we restrict our weights to $[0,1]$ this implies that the $m\ts{th}$ term should always be fully placed in the Trotter channel. The other piece of intuition is that smaller terms are pushed more towards the QDrift side of the partitioning. This can be seen from Eq. \eqref{eq:opt_first_order_weights} while considering the limit as $h_m \to 0$. If we assume that $\norm{[H_i, H_m]} \geq (1-w_i)$ on average, then the expression becomes $w_m \to -\infty$ in this limit, which we stop at 0.

One major drawback to the above expressions is the dependency of each optimal weight $w_m$ on every other weight $w_i$. As there does not seem
to be a clear basis in which to decouple these weights, this means that Eq. \eqref{eq:opt_first_order_weights} can only be used to update individual
weights given an initialization. This is the same situation as the greedy approach as discussed above, but we note that our expression gives us some
intuition as for why which weights or partitionings are chosen.

\subsection{Comparison with Trotter and QDrift}\label{sec:first_order_improvements}
Now that we have analyzed the cost and given a partitioning scheme we would like to know under what conditions this Composite channel can lead to comparable errors with lower gate cost. Instead of aiming to show that a Composite channel will outperform either first-order
Trotter or QDrift for arbitrary Hamiltonians we instead illustrate a concrete setting in which we achieve guaranteed asymptotic improvements. In 
later sections we are able to show more generic conditions on which asymptotic improvements can be obtained for higher-order formulas. 

The final case we consider for the first-order Trotter Composite channel is designed to take full advantage of this richer commutator structure of the Composite channel over first-order Trotter. Consider a Hamiltonian $H$ that has a partitioning into $A$ and $B$ such that the following two conditions hold
\begin{enumerate}
    \item The number of non-zero commutators between terms in $A$ scales with the square root of $L_A$. Mathematically, 
    \begin{equation}
        \abs{ \set{(i,j) : \norm{ [ A_i, A_j]} \neq 0} } \coloneqq N_{nz}^2 \in o(L_A).
    \end{equation}
    \item The strength of the $B$ terms, $\lambda_B = \sum_i b_i$, is asymptotically less than the the maximum commutator norm divided by the number of terms in $A$
    \begin{equation}\label{eq:lambdaBbd}
        \lambda_B^2 \leq a_{max}^2 \parens{{N_{nz}^4}/L_A^2},
    \end{equation}
    where $a_{max} \coloneqq \max_i a_i$.
    \item The number of terms in the $A$ partition is vanishingly small compared to the total number of terms $L_A \in o(L)$.
\end{enumerate}
Next we can use the optimal $N_B$ value from~\Cref{lem:first_order_opt_nb} and~\eqref{eq:lambdaBbd} to show that
\begin{equation}
    N_B^{-1}\in \bigo{\frac{1}{\lambda_B}\sqrt{\parens{N_{nz}^2 a_{max}^2+ L_A a_{max} \lambda_B}} }= \bigo{\frac{N_{nz} a_{max}}{\lambda_B} }.
\end{equation}
Similarly we have
\begin{equation}
    N_B \in \bigo{ \frac{\lambda_B\sqrt{L_A}}{a_{max} N_{nz} } }.
\end{equation}

Thus \Cref{thm:trotter_cost} shows that the number of exponentials needed to perform the simulation is in
\begin{align}
    C_{Comp}&\in \bigo{ \frac{t^2}{\epsilon}\left(L_A+\frac{\lambda_B \sqrt{L_A}}{a_{max} N_{nz} } \right) \left( a_{max}^2 N_{nz}^2 + L_A a_{max} \lambda_B + \frac{\lambda_B N_{n_z} a_{max}}{\sqrt{L_A}}\right) } \nonumber\\
    &\in \bigo{ \frac{t^2 L_A}{\epsilon} \left({a_{max}^2 N_{nz}^2}{} \right) } \\
    &\in o\parens{\frac{t^2}{\epsilon} L_A^2 a_{max}^2}.
\end{align}
Were we to use the lowest order Trotter formula for this simulation, the cost would be
\begin{align}
    C_{trot} &\in \bigo{\frac{t^2}{\epsilon} \left(L N_{nz}^2 a_{max}^2 \right) } \\
    &\in o\parens{\frac{t^2}{\epsilon} L L_A a_{max}^2} \\
    &\subseteq \omega(C_{comp}).
\end{align}
In contrast the cost for QDrift is
\begin{equation}
    C_{QD} = \bigo{ \frac{t^2}{\epsilon} \left({L^2 a_{max}^2} \right) } \subseteq \omega(C_{comp}).
\end{equation}
This shows that there exist circumstances where the cost of the Composite channel scales better than either of the two methods that compose it.


\section{Higher-Order Trotter Formulas} \label{sec:higher_order_trotter}
We now move on from first-order Trotter formulas to arbitrary higher-order Trotter formulas. To analyze this case there are a few distinct differences with the first-order channels. The first is that we now have a choice for what order formula we would like to use for the outer-loop decomposition of $\widetilde{\capU}$. For example, a first-order decomposition would be $\widetilde{\capU}(t) = \widetilde{\capUB}(t) \circ \widetilde{\capUA}(t)$ and a second-order decomposition would be $\widetilde{\capU}(t) = \widetilde{\capUA}(t/2) \circ \widetilde{\capUB}(t/2) \circ \widetilde{\capUB}(t/2) \circ \widetilde{\capUA}(t/2)$. In general, we can choose any order formula we like but it is analytically convenient to match the innermost Trotter formula. The next difference is that the time scaling between QDrift, Trotter, and the outer-loop errors could all be of different orders in $t/r$ which leads to a non-analytically solvable polynomial in $r$. The last issue that we address is that the commutator structure is no longer quadratic with respect to the Hamiltonian spectral norms, so we cannot follow the term weighting partitioning scheme from the first-order case. We will follow the same organizational structure as the first-order case and first set up our definitions and bound the diamond distance error, then compute the number of $e^{i H_i t}$ queries, followed by developing a partitioning scheme, and finally discuss the cost comparisons between our Composite channel and its constituents. 

\subsection{Query Complexity} \label{sec:higher_order_complexity}
We first need to determine an error equation for the Composite channel which we will then use to bound the number of iterations needed. In the first
order Trotter formula we simply used the following overall evolution $\widetilde{\capU} = \widetilde{\capUB} \circ \widetilde{\capUA}$, but this is
not sufficient for the higher-order case. We now introduce a generalization of this which mimics the Trotter formula recursion.
\begin{definition}[Higher-Order Outer-Loop]\label{def:higher_order_loop}
Let $\capU(t) = \capUB(t) \circ \capUA(t)$ denote the first-order outer-loop decomposition. We define the second-order outer-loop decomposition as
\begin{equation}
    \capU^{(2)}(t) \coloneqq \capUA(\nicefrac{t}{2}) \circ \capUB(\nicefrac{t}{2}) \circ \capUB(\nicefrac{t}{2}) \circ \capUA(\nicefrac{t}{2}).
\end{equation}
This forms the base case for the recursive strategy for higher-order outer-loops defined as
\begin{equation}
    \capU^{(2k)}(t) \coloneqq \capU^{(2k - 2)}(u_k t)^2 \circ \capU^{(2k - 2)}((1 - 4 u_k ) t) \circ \capU^{(2k - 2)}(u_k t)^2,
\end{equation}
where $u_k \coloneqq 1 / \parens{4-4^{1/(2k - 1)}}$ and $\Upsilon \coloneqq 2 \cdot 5^{k-1}$. Note that we use the same recursive strategy to define approximations
to the overall time evolution channel where we put tildes on each of the implemented channels.
\end{definition}

To analyze the overall error we need to break down the overall channel into individual channels that we have known results for. We specifically
use the approach of using the same decomposition order for the outer-loop that we use for the innermost Trotter formula.
\begin{lemma} \label{lem:diamond_dist_higher_order}
Let $\capU$ denote the exact unitary time evolution channel and $\widetilde{\capU}_{2k}$ denote an implemented product formula according to Definition \ref{def:higher_order_loop} for a partitioning of $H$ into $A, B$. Then we have the following diamond distance upper bound
\begin{equation}
    \diamondnorm{\capU(t) - \widetilde{\capU}^{(2k)}(t)} \leq \Upsilon \diamondnorm{\capUA(t) - \widetilde{\capUA}(t)} + \Upsilon \diamondnorm{\capUB(t) - \widetilde{\capUB}(t)}  + 2 \norm{e^{i H t} - \trotterU{2k}(\set{A,B}, t)}.
\end{equation}
\begin{proof}
The proof follows from repeated applications of the triangle inequality as well as subadditivity of the diamond norm with respect to channel composition
\begin{align}
    \diamondnorm{\capU(t) - \widetilde{\capU}^{(2k)}(t)} &\leq \diamondnorm{\capU(t) - \capU^{(2k)}(t)} + \diamondnorm{\capU^{(2k)}(t) - \widetilde{\capU}^{(2k)}(t)} \\
    &\leq 2 \norm{e^{i H t} - \trotterU{2k}(\set{A,B}, t)} + \diamondnorm{\capU^{(2k)}(t) - \widetilde{\capU}^{(2k)}(t)} \\
    &= 2 \norm{e^{i H t} - \trotterU{2k}(\set{A,B}, t)} \nonumber \\
    & ~ ~ + \diamondnorm{\capUB(t_\Upsilon) \circ \capUA(t_\Upsilon) \circ \ldots \capUB(t_1) \circ \capUA(t_1) - \widetilde{\capUB}(t_\Upsilon) \circ \widetilde{\capUA}(t_\Upsilon) \circ \ldots \widetilde{\capUB}(t_1) \circ \widetilde{\capUA}(t_1)} \\
    &\leq 2 \norm{e^{i H t} - \trotterU{2k}(\set{A,B}, t)} + \max_{t_i} \Upsilon \parens{\diamondnorm{\capUA(t_i) - \widetilde{\capUA}(t_i)} + \diamondnorm{\capUB(t_i) - \widetilde{\capUB}(t_i)}}.
\end{align} 
Note that each $t_i$ is a constant multiple of $t$, at each layer in the recursive formula $t$ picks up either a factor of $1 - 4 u_k$, $u_k$ or $\nicefrac{1}{2}$. Since $u_k \leq \nicefrac{1}{2}$ we can say that $\abs{1 - 4 u_k} \leq 1$. This implies that we can upper bound each time interval as $t_i \leq t$, which is sufficient for our purposes. Plugging this in to the previous equation yields the expression in the statement.  
\end{proof}
\end{lemma}

Now that we have derived a basis for the Composite channel error we can provide an upper bound on the number of operator exponentials needed to accurately approximate the ideal time evolution channel. 

\assCost*

\begin{proof}
We first note that by using prior arguments, namely Eq. \eqref{eq:err_to_err_per_iter} and Lemma \ref*{lem:diamond_dist_higher_order}, it is sufficient to show that $\diamondnorm{\capU(\nicefrac{t}{r}) - \widetilde{\capU}(\nicefrac{t}{r})} \leq \nicefrac{\epsilon}{r}$ to satisfy the total diamond distance error bound of $\epsilon$. Using Lemma \ref{lem:diamond_dist_higher_order} as well as the Trotter and QDrift errors from Eqs \eqref{eq:trotter_diamond_error} and \eqref{eq:qdrift_diamond_distance}
\begin{align}
    \diamondnorm{\capU(\nicefrac{t}{r}) - \widetilde{\capU}(\nicefrac{t}{r})} &\leq \Upsilon \diamondnorm{ \capUA(\nicefrac{t}{r}) - \widetilde{\capUA}(\nicefrac{t}{r})} + \Upsilon \diamondnorm{\capUB(\nicefrac{t}{r}) - \widetilde{\capUB}(\nicefrac{t}{r})} + 2\norm{e^{iHt/r} - \trotterU{2k}(\set{A,B}, \nicefrac{t}{r})}  \\
    &\leq 2 \Upsilon \norm{e^{i A t/r} - \trotterU{2k}(A,\nicefrac{t}{r})} + 2\norm{e^{iHt/r} - \trotterU{2k}(\set{A,B}, \nicefrac{t}{r})} + \parens{\frac{t}{r}}^2 \frac{4 \Upsilon \lambda_B^2}{N_B} \\
    &\leq \parens{\frac{t}{r}}^{2k+1} \frac{4 \Upsilon^{2k+1}}{2k + 1} \parens{\Upsilon \alpha_{comm}(A, 2k) + \alpha_{comm}(\set{A,B}, 2k)} + \parens{\frac{t}{r}}^2 \frac{4 \Upsilon \lambda_B^2}{N_B}.
\end{align}

It will prove useful for brevity to define the following quantities
\begin{align}
    P(t) &\coloneqq t^{2k+1} \frac{4 \Upsilon^{2k+1}}{2k+1}\parens{\Upsilon \alpha_{comm}(A,2k) + \alpha_{comm}(\set{A,B},2k)} \label{def:p_of_t} \\
    Q(t) &\coloneqq t^2 \frac{4 \Upsilon \lambda_B^2}{N_B} \label{def:q_of_t} \\
\end{align}
where $P$ represents the ``product formula" error and $Q$ captures the QDrift error. We can then use results from Theorems \ref{thm:trotter_cost} and \ref{thm:QDrift}, as well as the upper bound $\alpha_{comm}(\set{A,B}, 2k) \leq \Upsilon \alpha_{comm}(\set{A,B}, 2k)$ to write the following expressions
\begin{align}
    \frac{P(t)^{1/2k}}{\epsilon^{1/2k}} &\leq C_{Trott}(H, t, \epsilon) \frac{(1-q_B)^{1/2k}}{\Upsilon^{1 - 1/2k}L} \\
    \frac{Q(t)}{\epsilon} &= C_{QD}(H, t, \epsilon) \frac{\Upsilon \lambda_B^2}{\lambda^2 N_B}.
\end{align}
This gives our error as $\diamondnorm{\capU(\nicefrac{t}{r}) - \widetilde{\capU}(\nicefrac{t}{r})} \leq \frac{P(t)}{r^{2k+1}} + \frac{Q(t)}{r^2}$. We now pivot to finding a good value for $r$ that satisfies this inequality. Since there are no generic analytic solutions to polynomials of the form $a x^n + b x^2 = c$ for an arbitrary positive integer $n$ we have to resort to lower bounds on $r$. In other words, we would like to have a computable lower bound $r_{min} < r$ such that the following inequalities are satisfied
\begin{equation}
    \frac{P(t)}{r^{2k+1}} + \frac{Q(t)}{r^2} \leq \frac{P(t)}{r_{min}^{2k+1}} + \frac{Q(t)}{r_{min}^2} \leq \frac{\epsilon}{r} \leq \frac{\epsilon}{r_{min}}. \label{eq:big_rmin_bounds}
\end{equation}
We can make use of the above by finding expressions relating $r_{min}$ to $Q$ and $P$. The first expression we will make use of is
\begin{align}
    \frac{P(t)}{r^{2k+1}} + \frac{Q(t)}{r^2} \leq \frac{P(t)}{r^2 r_{min}^{2k-1}} + \frac{Q(t)}{r^2} &\leq \frac{\epsilon}{r} \\
    \frac{1}{\epsilon} \parens{\frac{P(t)}{r_{min}^{2k-1}} + Q(t)} &\leq r, \label{eq:intermediate_rmin_bound}
\end{align}
which reduces our task to finding a bound on $r_{min}$ using just $P$ alone. This is feasible if we revisit Eq. \ref{eq:big_rmin_bounds} and use the assumption that $Q(t) \geq 0$ for all possible inputs
\begin{align}
    \frac{P(t)}{r_{min}^{2k+1}} + \frac{Q(t)}{r_{min}^2} &\leq \frac{\epsilon}{r_{min}} \\
    \frac{P(t)}{r_{min}^{2k+1}} &\leq \frac{\epsilon}{r_{min}} \\
    \parens{\frac{P(t)}{\epsilon}}^{1/2k} &\leq r_{min}. \label{eq:rminBound}
\end{align}
Plugging \ref{eq:rminBound} into \ref{eq:intermediate_rmin_bound} yields
\begin{equation}
    \parens{\frac{P(t)}{\epsilon}}^{1/2k} + \frac{Q(t)}{\epsilon} \leq r.
\end{equation}

Now that we have a lower bound for $r$ this gives us an expression for the query cost as
\begin{align}
    C_{comp}(A,B, t, \epsilon) &= \Upsilon (\Upsilon L_A + N_B) r \\
    &= \Upsilon (\Upsilon L_A + N_B) \ceil{\frac{P(t)^{1/2k}}{\epsilon^{1/2k}} + \frac{Q(t)}{\epsilon}}, 
\end{align}
plugging in equations \eqref{def:p_of_t} and \eqref{def:p_of_t}, along with their simplifications in terms of $C_{Trott}$ and $C_{QD}$ 
yields the expressions in the statement.
\end{proof}

\subsection{Probabilistic Partitioning} \label{sec:probabilistic_partitioning}
As mentioned in Section \ref{sec:first_order_partitioning} there are multiple ways to go about determining a partition. In this section we develop 
a novel probabilistic approach that can be used to compute expectation values of necessary parameters in our composite cost function. Importantly the probabilities
are computable in time $\Theta(L)$ and require only simple constants to be evaluated. We first discuss why we cannot simply adapt the methods from
Section \ref*{sec:first_order_partitioning} before developing our approach. 

The main difficulty with adopting our weighting approach from before is the combinatorial landmine $\frac{\partial \alpha_{comm}(\set{A,B}, 2k)}{\partial w_m}$. When considering $\frac{\partial \alpha_{comm}(\set{A,B}, 2k)}{\partial w_m}$ note that $\alpha_{comm}$ has $2k+1$ factors of variously indexed $w_{\gamma}$ with nested commutators. Taking the derivative with respect to one of the innermost commutator terms leaves for a sum over all possibilities in where this term could be placed and over the remaining operators in the nested commutator. For feasibility one would have to upper bound $\alpha_{comm}$ with spectral 1-norms, such as $\alpha_{comm}(A, 2k) \in \bigo{\lambda_A^{2k+1}}$, which we do following this discussion. Unfortunately, these upper bounds do not help analytically as $\alpha_{comm}(\set{A,B}, 2k) \in \bigo{\sum_l \lambda_{A}^{l} \lambda_B^{2k+1-l}}$. The resulting derivative is a polynomial over all possible powers of $w_i$ and $(1-w_i)$, which as we have seen before does not have analytically solvable roots in general. This however could be used as the basis for a numeric approach, where one could compute these gradients as a subroutine in an optimization scheme, but it is not useful for our discussion.

The approach we consider is by reinterpreting the weights $w_i$ from Section \ref{sec:first_order_partitioning} as probabilities $p_i$ for each term to end up in Trotter or QDrift. This means that the expected Hamiltonian we simulate with Trotter is $\expect{A} = \sum_i p_i h_i H_i$ and the expected QDrift partition is $\expect{B} = \sum_i (1-p_i) h_i H_i$. We also introduce the indicators variables $I_i^A$ which is 1 if the $i\ts{th}$ term ends up in Trotter and 0 if it is in QDrift. Similarly we can define $I_i^B = 1 - I_i^A$. One main benefit is that this now gives us a probability over all possible partitions, which allows us to compute expectation values for quantities such as $C_{comp}, P,$ and $Q$. We remark that computing the expected value of the cost $\expect{C_{comp}(A,B)}$, which is our main priority, is different than computing the cost of the expected partition $C_{comp}(\expect{A}, \expect{B})$. The expected partition is computationally no different than the weighting scheme mentioned above so we instead compute the expectation of costs over partitions, which is clearly defined and computationally tractable.  

The first task we have is to find a useful distribution for each of the $p_i$'s. To do so we start with our cost function and introduce a heuristic that will allow us to find computable values for $p_i$. The first objects we introduce bounds for are the $\alpha_{comm}$ terms, which will be of additional use later on. 
\begin{lemma} \label{lem:bounds_on_alpha_and_p}
Let $\alpha_{comm}(\set{A,B}, 2k), \alpha_{comm}(A, 2k)$ be defined as in Eq. \eqref{def:alpha_comm}. Then it holds that 
\begin{align}
    \alpha_{comm}(\set{A,B}, 2k) &\leq 2^{2k} \sum_{l = 1}^{2k} \lambda_A^{l} \lambda_B^{2k + 1 - l} \\
    \alpha_{comm}(A, 2k) &\leq 2^{2k} \lambda_A^{2k+1}.
\end{align}
These inequalities can be used to upper bound $P(t)$ as
\begin{align}
    P(t) &\leq t^{2k+1} \frac{2^{2k + 2} \Upsilon^{2k+1}}{2k+1} \parens{\Upsilon \lambda_A^{2k+1} + \sum_{l = 1}^{2k} \lambda_A^{l} \lambda_B^{2k + 1 - l}} \\
    &\leq 2 \frac{2k + \Upsilon}{2k + 1} (2 \Upsilon \lambda t)^{2k+1} \eqqcolon P_{max}(t), \label{def:pMax}
\end{align}
where we introduce the upper bound $P_{max}(t)$ for $P$ which will be of use later.

\begin{proof}
The bounds on the $\alpha_{comm}$ factors are straightforward and computed as follows from the triangle inequality and submultiplicativity of the spectral norm. First we compute the commutator for just the $A$ terms
\begin{align}
    \alpha_{comm}(A, 2k) &\leq \sum_{\gamma_{2k+1}=1}^{L_A} \ldots \sum_{\gamma_1 =1}^{L_A} ||[A_{\gamma_{2k+1}}, \ldots,[A_{\gamma_2}, A_{\gamma_1}]\ldots] || \\
    &\leq \sum_{\gamma_{2k+1}=1}^{L_A} \ldots \sum_{\gamma_1 =1}^{L_A} 2^{2k} \norm{A_{\gamma_1}} \ldots \norm{A_{\gamma_{2k+1}}} \\
    &= 2^{2k} \lambda_A^{2k+1}, \label{eq:alphaCommA}
\end{align}
which is easily generalized to the commutator structure between the $A$ and $B$ terms as
\begin{align}
    \alpha_{comm}(\set{A,B},2k) &\leq \sum_{l=1}^{2k} \sum_{i_1, \ldots, i_{l} = 1}^{L_A} \sum_{i_{l+1},\ldots, i_{2k+1} = 1}^{L_B} \norm{[A_{i_1}, [A_{i_2},\ldots, [B_{i_{2k}}, B_{i_{2k+1}}]\ldots] } \\
    &\leq \sum_{l=1}^{2k} \sum_{i_1, \ldots, i_{l} = 1}^{L_A} \sum_{i_{l+1},\ldots, i_{2k+1} = 1}^{L_B} 2^{2k} \norm{A_{i_1}} \norm{A_{i_2}} \ldots \norm{B_{i_{2k}}} \norm{B_{i_{2k+1}}} \\
    &\leq 2^{2k} \sum_{l=1}^{2k} \lambda_A^{l} \lambda_B^{2k+1 - l}, \label{eq:alphaCommAB}
\end{align}
and we note that at least one power of $\lambda_A$ and $\lambda_B$ must be present in each term as there must be a minimum of one term from $A$ and a minimum of one term from $B$ in the original nested commutator. We can then use the simple bounds $\lambda_A \leq \lambda$ and $\lambda_B \leq \lambda$ to get partition independent bounds. Plugging these into the definition of $P(t)$ from Eq. \eqref{def:p_of_t} yields the bounds in the statement. 
\end{proof}
\end{lemma}

Returning to the original task of finding useful probabilities $p_i$, we now discuss what heuristics we can introduce to help us towards this goal. The cost of a Composite channel was computed as $C_{comp} \leq (\Upsilon L_A + N_B) \parens{\frac{P(t)^{1/2k}}{\epsilon^{1/2k}} + \frac{Q(t)}{\epsilon}}$, which we can use to give an upper bound on the expected cost as
\begin{align}
    \expect{C_{comp}(A, B, t, \epsilon, 2k)} &\leq \expect{(\Upsilon L_A + N_B) \parens{\frac{P(t)^{1/2k}}{\epsilon^{1/2k}} + \frac{Q(t)}{\epsilon}}} \\
    &\leq (\Upsilon L + N_B) \parens{\frac{\expect{P(t)^{1/2k} } }{\epsilon^{1/2k}} + \frac{\expect{Q(t)}}{\epsilon}},
\end{align}
where we upper bounded $L_A$ with $L$ and the number of gates performed during the QDrift channel as $N_B$. This latter point is a bit more subtle than it first appears, as there could be a non-zero probability of having 0 or 1 terms in $B$, which would mean no matter how many QDrift samples $N_B$ we take we only need to apply either 0 or 1 exponential gate to implement them. 

This expression makes clear that the Composite channel cost is a balancing act between higher-order product scaling and QDrift scaling. The heuristic we introduce is that we would like these expected quantities to be somewhat of the same magnitude. This is motivated by the observation that if one simulation method has much higher cost than another method we can simply start our partitioning off completely in the smaller cost method. We can then shift probability mass to the higher cost channel until the two contributions to the Composite channel are comparable. By making $\expect{P(t)^{1/2k}} \epsilon^{-1/2k} \approx \expect{Q(t)} \epsilon^{-1}$ rigorous we will get useful expressions for the probabilities $p_i$. This is done in the following lemma. 
\assProb*
\begin{proof}
We begin by making rigorous our notion that $\expect{P(t)^{1/2k}} \epsilon^{-1/2k} \approx \expect{Q(t)} \epsilon^{-1}$. To do so we will equate an upper bound for $\expect{P(t)}$ with a lower bound for $\expect{Q}$. As we know that $\expect{P(t)^{1/2k}} \leq P_{max}(t)^{1/2k}$ from Lemma \ref{lem:bounds_on_alpha_and_p}, we only need to lower bound $\expect{Q}$. This is rather straightforward
\begin{equation}
    \expect{Q(t)} = \frac{4 \Upsilon \expect{\lambda_B^2} t^2}{N_B} \geq \frac{4 \expect{\lambda_B}^2 t^2}{N_B}, 
\end{equation}
which is given from the definition of $Q$ along with Jensen's Inequality and convexity of $f(x) = x^2$. To enforce our heuristic of approximately equal we set the lower bound on $\expect{Q}/\epsilon$ to be less than the upper bound on $\expect{P^{1/2k}}/\epsilon^{1/2k}$. This is straightforward as
\begin{align}
    \frac{4 \Upsilon \expect{\lambda_B}^2 t^2}{N_B \epsilon} &\leq \frac{P_{max}(t)^{1/2k}}{\epsilon^{1/2k}} \\
    \expect{\lambda_B} &\leq \frac{1}{2t} \sqrt{\frac{P_{max}(t)^{1/2k} N_B \epsilon^{1-1/2k} } {\Upsilon} } .
\end{align}

As $\lambda_B$ is simply the sum of the spectral norms for each term in the $B$ channel, we can write it as $\sum_i \expect{I_i^B h_i}$. By plugging in the expectation of $I_i^B$ and $P_{max}$ from Eq. \eqref{def:pMax} into the above we arrive at the following
\begin{equation}
    \sum_i (1-p_i) h_i \leq \lambda  \sqrt{N_B \parens{\frac{\epsilon}{\lambda t}}^{1 - 1/2k} \parens{\frac{2k + \Upsilon}{2k+1}}^{1/2k} \frac{\Upsilon^{1/2k}}{2^{1-1/k}}  }. \label{eq:prob_bound_from_lambdab}
\end{equation}
It is now be apparent that we would like $1 - p_i$ to be proportional to the RHS of the above. There are a few adjustments that need to be made for consistency, such as introducing a minimum, dividing by $L$ and subtracting a factor of $\lambda$, and including these gives us our final definition as
\begin{equation}
    1 - p_i \coloneqq \min \set{\frac{\lambda }{h_i L} \parens{\sqrt{N_B \parens{\frac{\epsilon}{\lambda t}}^{1 - 1/2k} \parens{\frac{2k + \Upsilon}{2k+1}}^{1/2k} \frac{\Upsilon^{1/2k}}{2^{1-1/k}}  } - 1} , 1}.
\end{equation}
For convenience, we define 
\begin{equation}
    \chi \coloneqq \frac{\lambda }{L } \parens{\sqrt{N_B \parens{\frac{\epsilon}{\lambda t}}^{1 - 1/2k} \parens{\frac{2k + \Upsilon}{2k+1}}^{1/2k} \frac{\Upsilon^{1/2k}}{2^{1-1/k}}  } - 1 } \label{def:chi}
\end{equation}
such that $1 - p_i = \min \set{\frac{\chi}{h_i}, 1}$. We will also define
\begin{equation}
    \probIndexSet \coloneqq \set{i : 1 - p_i < 1} \label{def:omega}
\end{equation}
which consists of the indices of terms that have a non-zero probability of being in the Trotter channel. The complement of $\probIndexSet$, denoted $\probIndexSet^C$, contains the remaining indices of $\set{1, 2, \ldots, L} \ \probIndexSet$ that are guaranteed to be placed in the QDrift channel. We now show that the definition of our probabilities leads to the correct bound in \eqref{eq:prob_bound_from_lambdab}
\begin{align}
    \sum_i (1 - p_i) h_i &= \sum_{i \in \probIndexSet} \frac{\chi}{h_i} h_i + \sum_{i \in \probIndexSet^C} h_i \\
    &= \chi |\probIndexSet| + \lambda_{\probIndexSet^C} \\
    &= \frac{\lambda |\probIndexSet|}{L} \parens{\frac{1}{2}\parens{\frac{4k + 2 \Upsilon}{2k+1}}^{1/4k} \sqrt{N_B 2^{1 + 1/2k} \Upsilon^{1/2k} \parens{\frac{\epsilon}{\lambda t}}^{1 - 1/2k}} - 1} + \lambda_{\probIndexSet^C} \label{eq:consistency_bound_1} \\
    &\leq \frac{\lambda}{2}\parens{\frac{4k + 2 \Upsilon}{2k+1}}^{1/4k} \sqrt{N_B 2^{1 + 1/2k} \Upsilon^{1/2k} \parens{\frac{\epsilon}{\lambda t}}^{1 - 1/2k}} + (\lambda_{\probIndexSet^C} - \lambda) \label{eq:consistency_bound_2} \\
    &\leq \lambda \sqrt{N_B \parens{\frac{\epsilon}{\lambda t}}^{1 - 1/2k} \parens{\frac{2k + \Upsilon}{2k+1}}^{1/2k} \frac{\Upsilon^{1/2k}}{2^{1-1/k}}  }
\end{align}
as desired. In the step from \eqref{eq:consistency_bound_1} to \eqref{eq:consistency_bound_2} we assumed that the innermost term is positive. This is necessary for our probabilities to be greater than 0, but is not necessarily true as $\frac{\epsilon}{\lambda t} \to 0$ for fixed $N_B$. Therefore by requiring $1 - p_i > 0$ we introduce the following lower bound on $N_B$
\begin{equation}
    N_B \geq \parens{\frac{\lambda t}{\epsilon}}^{1 - 1/2k} \parens{\frac{2k + 1}{2k + \Upsilon}}^{1/2k} \frac{2^{1 - 1/k}}{\Upsilon^{1/2k}}. \label{eq:nb_lower_bound}
\end{equation}
We have therefore satisfied both guarantees as outlined in the statement of the Theorem, which completes the proof.
\end{proof}

There are a few comments to be made about the behavior of some of the quantities introduced in the above lemma. First we look at the lower bound on 
$N_B$, which scales as $\Theta\parens{\parens{\frac{\lambda t}{\epsilon}}^{1 - 1/2k}}$. If we assume that $t, \epsilon$ are independent of $L$ then this overall scaling is sublinear with respect to $L$ as $\lambda \leq \max_i h_i L$, which indicates that by moving a term from Trotter to QDrift we do not automatically lose out in gate cost. If $t, \epsilon$ are dependent on $L$ then we cannot make the same guarantee. Second, if we parametrize $N_B$ to be within a constant factor of this lower bound $N_B = ( 1 + c)^2 \parens{\frac{\lambda t}{\epsilon}}^{1 - 1/2k} \parens{\frac{2k + 1}{2k + \Upsilon}}^{1/2k} \frac{2^{1 - 1/k}}{\Upsilon^{1/2k}}$, we can then simplify the expression for $\chi$ as 
\begin{align}
    \chi &= \frac{\lambda }{L} \parens{\sqrt{N_B \parens{\frac{\epsilon}{\lambda t}}^{1 - 1/2k} \parens{\frac{2k + \Upsilon}{2k+1}}^{1/2k} \frac{\Upsilon^{1/2k}}{2^{1-1/k}}  } - 1} \\
    &= c \frac{\lambda}{L}.
\end{align}
This is a nice simplification that we will use later and shows $\chi$ can be thought of as an ``average strength" of the overall Hamiltonian. This then gives the intuition that our probability definitions are somewhat analogous to an inverse importance sampling procedure with respect to the spectral norms $h_i$. We note that as $c \to \infty$, $1-p_i = \min \set{\frac{c \lambda}{L h_i} , 1} \to 1$, which implies that $p_i \to 0$. This means that as we increase the number of QDrift samples the probability distribution will put more probability mass into the QDrift partition. Intuitively, the resulting partition takes advantage of having a low-error QDrift simulation by placing more terms into it's partition. Contrast this behavior with the opposite limit, $c \to 0$. In this case we have a very noisy QDrift partition, due to the QDrift error scaling as $1/N_B$, and we see that the distribution in this case properly places more probability mass into the Trotter partition.


\subsection{Comparison with Trotter and QDrift}

In this section we shift our focus to analyzing when a two term Composite channel can outperform a simulation of just a Trotter or QDrift channel. 
The first result we show gives asymptotic bounds on certain quantities that result from a partitioning, such as $\lambda_B$, that yield asymptotic
improvements for the overall Composite channel query cost over either Trotter or QDrift. After analyzing the asymptotic cost for a 
predetermined partition we then investigate when the probabilistic partitioning scheme introduced above can yield gate cost improvements. We find that
in the extremal cases for the probabilities (i.e $p_i \to 0$ or $p_i \to 1$ for all $i$)
our resulting Composite channel exactly matches the QDrift and second-order Trotter costs, with minor constant factors for higher-order Trotter
channels. Finally, in the case where a Hamiltonian has spectral norms that decay exponentially, i.e. $h_i = 2^{-i}$, in expectation our partitioning
meets the requirements for the asymptotic improvements shown in Theorem \ref*{thm:higher_order_improvements_general} under the condition that
$t/\epsilon$ is such that $C_{Trott} = C_{QD}$. 

\subsubsection{Deterministic Partitioning Improvements} \label{sec:higher_order_improvements}

We first investigate the asymptotic bounds that a partition must satisfy in order to offer asymptotic improvements over either Trotter or QDrift. 

\assImprovements*

\begin{proof}
    We start with the expression for the Composite channel cost from Theorem \ref{thm:higher_order_cost_fixed}    
    \begin{align}
         C_{comp} \leq \Upsilon (\Upsilon L_A + N_B) \ceil{C_{Trott}(H, t, \epsilon, 2k) \frac{(1-q_B)^{1/2k}}{\Upsilon^{1-1/2k} L} + C_{QD}(H, t, \epsilon) \frac{\Upsilon}{N_B} \frac{\lambda_B^2}{\lambda^2}}, \label{eq:cost_comparison_intermediate} 
    \end{align}
    where $1-q_B$ quantifies the contribution of the $A$ partition to the overall nested commutator structure $\alpha_{comm}(H,2k)$. The most 
    straightforward way to proceed is to split this expression based on our two cases, $\beta > 1$ and $0 < \beta < 1$. 
    
    We first examine
    the case when $\beta > 1$, which implies $C_{Trott} < C_{QD}$. Our composite cost expression then can be written as
    \begin{equation}
        C_{comp} \leq C_{Trott} \parens{\Upsilon^{1 +1/2k} (1-q_B)^{1/2k} \frac{L_A}{L} + \Upsilon^3 \frac{\lambda_B^2}{\lambda^2} \frac{L_A}{N_B} C_{QD}^{1 - 1/\beta} + \Upsilon^{1/2k} (1-q_B)^{1/2k} \frac{N_B}{L} + \Upsilon^2 \frac{\lambda_B^2}{\lambda^2} C_{QD}^{1-1/\beta} } .
    \end{equation}
    Our goal is to show that if each of the terms above
    are in $o(1)$, then the sum is in $o(1)$ which implies $C_{comp} \in o(C_{Trott})$. Starting with the simplest term 
    \begin{equation}
        \Upsilon^{1 + 1/2k} \frac{(1-q_B)^{1/2k}L_A}{L} \in o(1),
    \end{equation}
    which holds when $L_A (1 - q_B)^{1/2k} \in o(L)$. The second term we analyze is the QDrift only term 
    \begin{equation}
        \Upsilon^2 \frac{\lambda_B^2}{\lambda^2} C_{QD}^{1-1/\beta}  = \parens{\frac{t^2}{\epsilon}}^{1-1/\beta} \frac{\lambda_B^2}{\lambda^{1/\beta}} \in o(1),
    \end{equation}
    where we used the assumption that $\lambda_B \in o\parens{\lambda^{1/2\beta} \parens{\frac{\sqrt{\epsilon}}{t}}^{1-1/\beta}}$ to reduce the expression to $o(1)$. We now move on to the remaining terms involving $N_B$. The first we can simplify is 
    \begin{equation}
        \Upsilon^{1/2k} (1-q_B)^{1/2k} \frac{N_B}{L} \in o(1),
    \end{equation}
    by the assumption that $N_B (1 - q_B)^{1/2k} \in o(L)$. The second term involving $N_B$ is
    \begin{equation}
        \Upsilon^3 \frac{ C_{QD}^{1 - 1/\beta} \lambda_B^2}{\lambda^2} \frac{L_A}{N_B} \in o(1),
    \end{equation}
    where we used the fact that $\frac{ C_{QD}^{1 - 1/\beta} \lambda_B^2}{\lambda^2} \in o(1)$ and the assumption that $N_B \in \Omega(L_A)$. Given that all four terms in the cost function expression \ref{eq:cost_comparison_intermediate} are $o(1)$, we have shown that $C_{comp} \in o(C_{Trott})$ for $\beta > 1$. 
    
    We can now move on to the case when $0 < \beta < 1$, which essentially repeats the above logic. In this situation the cost expression from Eq \eqref{eq:cost_comparison_intermediate} reduces to 
    \begin{equation}
        C_{comp} \leq C_{QD} \parens{\Upsilon^{1 +1/2k} \frac{(1-q_B)^{1/2k} C_{Trott}^{1 - 1/\beta} L_A}{L} + \Upsilon^3 \frac{\lambda_B^2}{\lambda^2} \frac{L_A }{N_B} + \Upsilon^{1/2k}  (1-q_B)^{1/2k} \frac{N_B }{L} C_{Trott}^{1-1/\beta} + \Upsilon^2 \frac{\lambda_B^2}{\lambda^2}}.
    \end{equation}
    Starting with the rightmost term we have
    \begin{equation}
        \Upsilon^2 \frac{\lambda_B^2}{\lambda^2} \in o(1),
    \end{equation}
    which is guaranteed by the assumption $\lambda_B \in o(\lambda)$. We then look at the only other term that does not include a factor of $N_B$,
    which reduces to
    \begin{equation}
        \Upsilon^{1 +1/2k} (1-q_B)^{1/2k}\frac{L_A}{L} C_{Trott}^{1 - 1/\beta} = \Upsilon^{2 + 1/2k} \parens{\frac{t^{1+1/2k}}{\epsilon^{1/2k}}}^{1 - 1/\beta} \frac{L_A}{L^{1/\beta}} \parens{\frac{\alpha_{comm}(A, 2k) + \alpha_{comm}(\set{A,B}, 2k)}{\alpha_{comm}^{1/\beta}(H, 2k)}}^{1/2k}.
    \end{equation}
    This can be shown to be in $o(1)$ given Assumption 1 for $ 0 < \beta < 1$. We can utilize the fact that this term is vanishing to reduce one of the other terms involving $C_{Trott}$ as 
    \begin{equation}
        \Upsilon^{1/2k} N_B \frac{(1-q_B)^{1/2k} C_{Trott}^{1-1/\beta}}{L} = \frac{N_B}{L_A} \cdot o(1) \in o(1),
    \end{equation}
    where we use the assumption that $N_B \in \bigo{L_A}$ for the last step. The last remaining term is rather straightforward
    \begin{equation}
        \Upsilon^2 \frac{\lambda_B^2}{\lambda^2} \frac{L_A}{N_B} = \frac{L_A}{N_B} \cdot o(1) \in o(1),
    \end{equation}
    where we used $N_B \in \Omega(L_A)$ as well as  $\lambda_B \in o( \lambda)$. Taken together, we have shown that $C_{comp} = C_{QD} \cdot o(1)$, implying $C_{comp} \in o(C_{QD})$ for $0 < \beta < 1$.

    This combined with the above proof for $\beta > 1$ implies that $C_{comp} \in o(\min \set{C_{QD}, C_{Trott}})$ if the assumptions in the
    statement are met, completing the proof.
\end{proof}

\begin{table}[h!]
\centering
\begin{tabular}{|c||c|c|}
    \hline
    & $C_{QD} > C_{Trott}$ & $C_{QD} < C_{Trott}$ \\
    \hline \hline
    $L_A \in$ & $o\parens{\frac{L}{(1-q_B)^{1/2k}}}$ & $o\parens{ L^{1/\beta} \parens{\frac{\epsilon^{1-1/\beta}}{t^{(2k+1)(1-1/\beta)}} \frac{\alpha_{comm}^{1/\beta}(H)}{\alpha_{comm}(A) + \alpha_{comm}(\set{A,B})} }^{1/2k}}$ \\
    \hline $\lambda_B \in$ & $o \parens{\lambda^{1/\beta} \parens{\frac{\sqrt{\epsilon}}{t}}^{1-1/\beta}}$ & $ o(\lambda)$ \\
    \hline (Lower Bound) $N_B \in$ & $\Omega(L_A)$ & $\Omega(L_A)$ \\
    \hline (Upper Bound) $N_B \in$ & $o\parens{\frac{L}{(1-q_B)^{1/2k}}}$ & $\bigo{L_A}$ \\
    \hline
\end{tabular}
\caption{Summary of asymptotic requirements for parameters of interest when $C_{QD} = C_{Trott}^{\beta}$ to yield $C_{Comp} \in o(\min \set{C_{QD}, C_{Trott}})$.}
\end{table}

Now that we have given bounds on a partitioning that can yield asymptotic improvements, we turn to the problem of finding a good value of $N_B$ that 
can satisfy these necessary assumptions. We first show a straightforward generalization of the first-order optimal value of $N_B$ to the higher-order Trotter case. 
\begin{lemma} \label{lem:optimal_nb_higher_order}
    Given a Hamiltonian $H$ with fixed partitions $A$ and $B$, the optimal number of QDrift samples $N_B$ is given as
    \begin{equation}
        N_B = 2 \lambda_B \sqrt{\parens{\frac{\Upsilon t}{\epsilon}}^{1-1/2k} L_A \parens{\frac{2k+1}{\Upsilon \alpha_{comm}(A) + \alpha_{comm}(\set{A,B})}}^{1/2k}}
    \end{equation}
    \end{lemma}
\begin{proof}
    This is a rather straightforward result which follows the logic of the first-order results in section \ref{sec:first_order_improvements} so we will simply show some of the intermediate steps. Starting with the upper bound on $C_{Comp}$
    \begin{align}
        C &\leq \Upsilon (\Upsilon L_A + N_B) \parens{\frac{P(t)^{1/2k}}{\epsilon^{1/2k}} + \frac{Q(t)}{\epsilon}} ,\end{align}
    we then require that the upper bound be an optima with respect to $N_B$, namely that $\frac{\partial}{\partial N_B} C_{Comp} = 0$. Computing the derivative using the expressions for $Q(t)$ and $P(t)$, given in Eqs. \eqref{def:q_of_t} and \eqref{def:p_of_t} respectively, we find
        \begin{align}
        0 &= \Upsilon \parens{\frac{P(t)^{1/2k}}{\epsilon^{1/2k}} + \frac{Q(t)}{\epsilon}} - \Upsilon(\Upsilon L_A + N_B) \frac{Q(t)}{N_B^2} \\
        \frac{P(t)^{1/2k}}{\epsilon^{1/2k}} N_B^2 &= \Upsilon L_A \frac{4 \lambda_B^2 t^2}{\epsilon} \\
        N_B &= 2 \lambda_B \sqrt{\frac{\Upsilon L_A t^2}{\epsilon^{1-1/2k} P(t)^{1/2k}}}.
    \end{align}
    Plugging in the expression for $P(t)$ yields the stated result.
\end{proof}

Now that we have given an expression for an optimal value of $N_B$ given a partitioning, we briefly show that this $N_B$ value
satisfies the assumptions from Theorem \ref{thm:higher_order_improvements_general} for the case when $C_{QD} = C_{Trott}$. As mentioned
in the theorem statement, for this situation we note that $N_B \in \omega(L_A)$ and $N_B \in o(L)$ is sufficient to satisfy the requirements.
The equal cost condition yields the equation
\begin{align}
    C_{QD} &= C_{Trott} \\
    \parens{\frac{t}{\epsilon}}^{1-1/2k} &=  \frac{L \alpha_{comm}^{1/2k}(H, 2k)}{\lambda^2} \cdot O(1).
\end{align}
This then simplifies the optimal value for $N_B$ as
\begin{align}
    N_B &= 2 \lambda_B \sqrt{\parens{\frac{\Upsilon t}{\epsilon}}^{1-1/2k} L_A \parens{\frac{2k+1}{\Upsilon \alpha_{comm}(A, 2k) + \alpha_{comm}(\set{A,B}, 2k )}}^{1/2k}} \\
    &= \frac{\lambda_B}{\lambda} \sqrt{L_A L} (1 - q_B)^{1/2k} \cdot O(1).
\end{align}
This last expression is clearly seen to lie in both $\omega(L_A)$ and $o(L)$ if $\lambda_B \in o(\lambda)$ and $L_A \in o(L)$. This shows that 
at times $t$ and errors $\epsilon$ where QDrift and Trotter have equal query cost the optimal value of $N_B$ leads to a Composite channel that has
asymptotic improvements over either of the channels it composes.

\subsubsection{Probabilistic Partitioning Performance} \label{sec:prob_limits}
We now focus on showing that a Composite channel with a probabilistic partitioning scheme according to Lemma \ref{lem:prob_lemma} can lead to gate cost savings compared to either Trotter or QDrift channels alone. The main expression we will be working with is
\begin{align}
    \expect{C_{comp}(H, t, \epsilon, N_B)} &\leq \expect{(L_A + N_B)\parens{\frac{P(t)^{1/2k}}{\epsilon^{1/2k}} + \frac{Q(t)}{\epsilon}}} \\
    &\leq \sqrt{\expect{L_A^2} \frac{\expect{P(t)^{1/k}}}{\epsilon^{1/k}}} + \sqrt{\expect{N_B^2} \frac{\expect{P(t)^{1/k}} }{\epsilon^{1/k}} } \nonumber \\
    &\quad + \sqrt{\expect{L_A^2} \frac{\expect{Q(t)^2}}{\epsilon^2}} + \sqrt{\expect{N_B^2} \frac{\expect{Q(t)^2}}{\epsilon^2}} \\
    &\leq \parens{\sqrt{\expect{L_A^2}} + \sqrt{\expect{N_B^2}}} \parens{\frac{\expect{P(t)^{1/2k}}}{\epsilon^{1/2k}} + \frac{\sqrt{\expect{Q(t)^2}}}{\epsilon}}, \label{eq:expected_cost_basic}
\end{align}
which follows from repeated applications of the Cauchy-Schwarz inequality and Jensen's inequality for $f(x) = x^{1/k}$ for positive $x$. The approach we take is to first give useful upper bounds on this expression and then show that these bounds saturate to the costs for Trotter and QDrift as we 
shift our probability mass towards either 0 or 1 for each term. Finally, we investigate a specific Hamiltonian that in expectation satisfies the
assumptions for asymptotic improvements from Theorem \ref{thm:higher_order_improvements_general}. 

In this section we will present some lemmas that provide useful upper bounds to the quantities used in the above expected cost expression. Since these proofs are relatively straightforward applications of the definitions given above we leave the proofs to appendix \ref{sec:appendix_a}. 

\begin{restatable}{lemma}{laSquared} \label{lem:la_squared}
Let $L_A$ denote the number of terms simulated with a higher-order Trotter formula in a Composite channel with the assumptions as defined in section \ref{sec:prelim}. Using a probabilistic partitioning scheme as defined in \ref{lem:prob_lemma}, then the second moment of $L_A$ obeys
$$\mathbb{E}(L_A^2) \le |\probIndexSet|^2 - |\probIndexSet| \sum_{i \in \probIndexSet} \frac{\chi}{h_i},$$
where $\chi$ is given in \ref{def:chi} and $\probIndexSet$ is the set of $p_i$ such that $1-p_i < 1$ as given in \ref{def:omega}.
\end{restatable}
This result makes intuitive sense, as $L_A$ can be no larger than $|\probIndexSet|$ as this is the set of terms that have non-zero probability of ending up in the Trotter partition. This expression then adjusts this upper bound by $|\probIndexSet| \sum_{i \in \probIndexSet} (1-p_i)$, which can be thought of as the expected number of terms that will end up being placed in QDrift from $\probIndexSet$. 

We know move on to bounding the expected contribution from the QDrift channel due to the error, namely $Q(t)$. 
\begin{restatable}{lemma}{qUpperBounds} \label{lem:expect_q_upper_bounds}
Let $Q$ denote the contribution to the error of a Composite channel, where $Q$ is defined in \ref{def:q_of_t}, with the standard assumptions from Section \ref{sec:prelim} and a probabilistic partitioning scheme as outlined in \ref{lem:prob_lemma}. The following upper bounds hold on the first two moments for $Q$
\begin{align}
    \expect{Q(t)} &\leq \frac{4 t^2}{N_B} \parens{\chi \lambda_{\probIndexSet} + \parens{\chi |\probIndexSet| + \lambda_{\probIndexSet^C}}^2} \label{eq:expect_q_upper_bound} \\
    \expect{Q(t)^2} &\leq \frac{4 t^2 \lambda^2}{N_B} \expect{Q(t)} \leq \frac{16 t^4 \lambda^2}{N_B^2} \parens{\chi \lambda_{\probIndexSet} + \parens{\chi |\probIndexSet| + \lambda_{\probIndexSet^C}}^2} \label{eq:expect_q_squared_upper_bound}
\end{align}
\end{restatable}

The main intuition to be gained from here is that as $|\probIndexSet| \to 0$ the dominant contribution comes from $\lambda_{\probIndexSet^C}$ which is to be expected as this is the set of terms that are guaranteed to be placed in QDrift.

We now move on to the last term we need to bound, which is $P$. 
\begin{restatable}{lemma}{pUpperBound} \label{lem:p_upper_bound}
Let $P(t)$ denote the product formula error scaling as defined in \ref{def:p_of_t} for a Composite channel with the standard assumptions as defined in \ref{sec:prelim}. Using a probabilistic partitioning scheme as defined in \ref{lem:prob_lemma} then the following upper bounds hold for the expected value of $P(t)$
\begin{equation}
    \expect{P(t)} \leq \frac{(2 \Upsilon)^{2 + 2k}}{2k+1} t^{2k+1} \lambda^{2k} \parens{\lambda_{\probIndexSet} - \chi |\probIndexSet|} .\label{eq:p_upper_bound}
\end{equation}
\end{restatable}
Even though this upper bound on $P$ is not tight at all we can still capture interesting edge case behavior. Note that as $\probIndexSet \to \varnothing$ this indicates that all of our terms are in the set $\probIndexSet^C$, meaning the probability they are in the Product Formula channel tends to 0, which is reflected as $\lambda_{\probIndexSet} \to 0$ and $|\probIndexSet| \to 0$ which implies that $\expect{P(t)} \to 0$ appropriately. We will investigate the other 
regime in which $1 - p_i \to 0$ shortly. 

Now that we have workable upper bounds we would like to make some direct comparisons between a Composite channel and both Trotter and QDrift. However before showing that this expression can outperform known upper bounds on Trotter or QDrift, we first would like to perform some consistency checks. What we need to show is that as we shift the probability mass over our partitions, the Composite channel cost should tend towards Trotter or QDrift, based on which direction we shift the mass. In other words, as $p_i \to 0$ for all $i$, we would like $\expect{C_{Comp}} \to C_{QD}$ and as $p_i \to 1$ we would like $\expect{C_{Comp}} \to  C_{Trott}$. In the following Theorem we verify this intuition precisely for the QDrift regime and second-order Trotter formulas. For higher-order formulas we are able to show that as $p_i \to 1$ the Composite channel cost is at most $1.12$ times worse than the Trotter cost.
 
\begin{theorem}[Expected Cost reduction to Trotter and QDrift in proper limits]
    Assume as inputs a Hamiltonian $H$, time $t$, and error $\epsilon$. Let $p_i$ denote the probability of assigning term $H_i$ of $H$ to the Trotter partition of a composite simulation channel, given by Lemma \ref{lem:prob_lemma}. Then as we vary the number of QDrift samples $N_B$ the expected cost of the Composite channel, $\expect{C_{Comp}}$, saturates towards the Trotter and QDrift gate costs in the respective limits for $N_B$. Specifically, as $N_B$ approaches its lower bound in Eq. \eqref{eq:nb_lower_bound} then $\expect{C_{Comp}} \to \parens{\Upsilon^{1/2k} / 2^{1 - 1/2k}} \cdot C_{Trott} \leq 1.12 C_{Trott}$ and as $N_B \to \infty$ then $\expect{C_{Comp}} \to C_{QD}$ exactly. Note that for $k=1$, or a second-order Trotter formula, $\Upsilon^{1/2k} / 2^{1-1/2k} = 1$.
\end{theorem}
\begin{proof}
The first limit we tackle is the QDrift regime. Dropping the specific probabilities that were defined in Lemma \ref{lem:prob_lemma} we consider adjusting each probability $p_i \to 0$ for all $i$. We have three pieces we then need to upper bound: $P$, $Q$, and $L_A$. We can ignore $\sqrt{\expect{N_B^2}}$ as there is no probability we do not have an empty or near-empty QDrift channel. The first observation is that as each $p_i \to 0$ and as we set more and more probabilities to $1 - p_i = 1$ then the size of $|\probIndexSet|$ will decrease accordingly. We consider the case where $|\probIndexSet| \to 0$ as this represents the QDrift only channel. This limit then implies that $\lambda_{\probIndexSet} \to 0$ and $\lambda_{\probIndexSet^C} \to \lambda$. This allows us to vastly simplify many of our bounds
\begin{align}
    \expect{L_A^2} &\leq |\probIndexSet|^2 - |\probIndexSet| \sum_{i \in \probIndexSet} \frac{\chi}{h_i} \to 0 \\
    \expect{P(t)^{1/2k}} & \frac{(2 \Upsilon)^{1 + 1/k}}{(2k+1)^{1/2k}} \lambda (\lambda_{\probIndexSet} - \chi |\probIndexSet|)^{1/2k} \to 0 \\
    \expect{Q(t)^2} &\leq \frac{16 t^4 \lambda^2}{N_B^2} \parens{\chi \lambda_{\probIndexSet} + \parens{\chi |\probIndexSet| + \lambda_{\probIndexSet^C}}^2} \to \frac{16 t^4 \lambda^4}{N_B^2}.
\end{align}
By plugging these into the cost expression from Eq. \ref{eq:expected_cost_basic} we get the following limit
\begin{align}
    \expect{C_{comp}(H, t, \epsilon, N_B) | p_i \to 0 \text{ } \forall i} &\leq \parens{\sqrt{\expect{L_A^2}} + \sqrt{\expect{N_B^2}}} \parens{\frac{\expect{P(t)^{1/2k}}}{\epsilon^{1/2k}} + \frac{\sqrt{\expect{Q(t)^2}}}{\epsilon}} \\
    &\to \frac{4 t^2 \lambda^2}{\epsilon},
\end{align}
which exactly matches the upper bounds on the cost of a QDrift only channel. We remind the reader of the extraneous factor of 2 that we incur from restricting $\epsilon \in (0, \ln(2) \lambda t)$ as opposed to Campbell's original result \cite{qdrift}. 

Next we look at the limit in which $p_i \to 1$ and we are performing a Trotter only channel. By using the probability distributions from Lemma \ref{lem:prob_lemma} it is most straightforward to consider this limit by using a parametrized value for $N_B$ 
\begin{equation}
    N_B(c) = (1 + c)^2 \parens{\frac{\lambda t}{\epsilon}}^{1 - 1/2k} \frac{1}{4(2 \Upsilon)^{1 + 1/2k}} \parens{\frac{2k + 1}{4k + 2 \Upsilon}}^{1/2k},
\end{equation}
which was shown to lead to $\chi = c \frac{\lambda}{L}$. For probabilities which are in $\probIndexSet$, we can say $p_i = 1 - c \frac{\lambda}{L h_i}$. This then implies that by sending $c \to 0$ we appropriately have $\chi \to 0$ and $p_i \to 1$. We can use this to compute the limiting upper bound
on $\expect{L_A^2}$ as 
\begin{equation}
    \expect{L_A^2} \leq |\probIndexSet|^2 - |\probIndexSet| \sum_i \nicefrac{\chi}{h_i} \to L^2. 
\end{equation}
The contribution due to the QDrift error is also straightforward in the limiting case as $c \to 0$
\begin{equation}
    \expect{Q(t)^2} \leq \frac{16 t^4 \lambda^2}{N_B^2} \parens{\chi \lambda_{\probIndexSet} + \parens{\chi |\probIndexSet| + \lambda_{\probIndexSet^C}}^2} \to 0,
\end{equation}
as $\chi \to 0$ and $\lambda_{\probIndexSet^C} \to 0$. The remaining simple upper bound is for $P$ which is as follows
\begin{equation}
    \expect{P(t)^{1/2k}} \leq \frac{(2 \Upsilon)^{1+1/k}}{(2k+1)^{1/2k}} t^{1+1/2k} \lambda (\lambda_{\probIndexSet} - \chi |\probIndexSet|)^{1/2k} \to \frac{(2 \Upsilon)^{1+1/k}}{(2k+1)^{1/2k}} (\lambda t)^{1+1/2k}.
\end{equation}
The last quantity we have to bound is $\sqrt{\expect{N_B^2}}$, which up until now has been harmlessly upper bounded by $N_B$. However we now have a very high probability, tending to 1, of not having any terms in our QDrift channel. This means that the expected number of exponential gates performed during the QDrift sampling procedures should tend to 0. We make this rigorous through a simple definition of expectation
\begin{align}
    \expect{N_B^2} &= \prob{|B| = 0} \cdot 0 + \prob{|B| = 1} \cdot 1 + \prob{|B| > 1 } \cdot N_B^2 \\
    &= \sum_i \prod_{j \neq i} (1 - p_i) p_j + (1 - \prob{|B| = 0} - \prod_{j \neq i} (1 - p_i) p_j) \cdot N_B^2 \\
    &\to 0 + (1 - \prod_i p_i) \cdot N_B^2 \\
    &\to 0.
\end{align}
Intuitively, the above is simply a result of $\prob{|B| = 0} \to 1$ as $p_i \to 1$, resulting in all other terms tending to 0. This gives the final upper bound on our expected cost as
\begin{align}
    \expect{C_{comp}(H, t, \epsilon, N_B) | p_i \to 1 \text{ } \forall i} &\leq \parens{\sqrt{\expect{L_A^2}} + \sqrt{\expect{N_B^2}}} \parens{\frac{\expect{P(t)^{1/2k}}}{\epsilon^{1/2k}} + \frac{\sqrt{\expect{Q(t)^2}}}{\epsilon}} \\
    &\to \frac{(L \lambda t)^{1 + 1/2k}}{\epsilon^{1/2k}} \frac{(2 \Upsilon)^{1+1/k}}{(2k+1)^{1/2k}}.
\end{align}
We note that by using similar upper bounds on the factor of $\alpha_{comm}$ for the Trotter gate cost yields the following
\begin{align}
    C_{Trott} &\leq \frac{L t^{1+1/2k}}{\epsilon^{1/2k}} \alpha_{comm}^{1/2k}(H, 2k) \frac{\Upsilon^{1 + 1/2k} 4^{1/2k}}{(2k+1)^{1/2k}} \\
    &= \frac{L (\lambda t)^{1+1/2k}}{\epsilon^{1/2k}} \frac{\Upsilon^{1+1/2k} 2^{2 + 1/2k}}{(2k+1)^{1/2k}}.
\end{align}
The simple ratio $\expect{C_{comp}} / C_{trott} = \Upsilon^{1/2k} 2^{-1 + 1/2k} = 2^{-1 + 1/k} 5^{1/2 - 1/2k}$ shows that for $k = 1$ we exactly
match the Trotter cost. For higher-orders this constant factor is negligible and is upper bounded by $1.12$. 
\end{proof}

\subsubsection{Exponentially Decaying Hamiltonian Terms}

Now that we have shown that the probabilistic partitioning scheme leads to a composite cost that saturates to $C_{Trott}$ and $C_{QD}$ in the appropriate regime, we will analyze a situation in the middle of these two limits. This is rather difficult to do in generality, so we will restrict our attention to Hamiltonians that have exponentially decaying spectral norms for each term, i.e. $h_i = 2^{-i}$ when sorted by spectral norm, and at the crossover time $t$ such that $C_{QD} = C_{Trott}$. For this analysis we would like to show that the parameters of the partitioning, namely $L_A$, $N_B$, and $\lambda_B$, satisfy the constraints for asymptotic improvements given by Theorem \ref{thm:higher_order_improvements_general}. At the crossover time these constraints are simply $L_A \in o(L)$, $\lambda_B \in o(\lambda)$ and $N_B \in \Theta(L_A)$. We prove that in expectation these parameters will satisfy these constraints, then by using simple tail bounds we show that the probability the parameters fall into an asymptotically bad regime is vanishingly small.

\begin{theorem} \label{thm:exponential_decay}
    Given a Hamiltonian $H = \sum_i h_i H_i$ with exponentially decaying spectral norms $h_i = 2^{-i}$, a small commutator structure $\alpha_{comm} \in \Theta \parens{\frac{\log_2 L}{L}}$, time $t$, and error $\epsilon$ such that $C_{Trott}(H, t, \epsilon) = C_{QD}(H,t,\epsilon)$, then the probabilistic partitioning scheme given by Lemma \ref{lem:prob_lemma} yields a partition that satisfies the asymptotic requirements of Theorem \ref{thm:higher_order_improvements_general} with high probability.
\end{theorem}
\begin{proof}
    The outline of the proof is to first provide a parametrization of $N_B$ that simplifies the calculation, then use this to bound expectations of $L_A$ and $\lambda_B$, then finally to use tail bounds to show that the expectation values are sufficient to work with. We first resolve a conundrum involving $N_B$. As we have seen in Section \ref{sec:prob_limits}, the value of $N_B$ essentially determines the partitioning, however in Theorem \ref{thm:higher_order_improvements_general} we see that $N_B$ has to satisfy certain asymptotic requirements that depend on the partitioning! To get around this we introduce the following parametrization
    \begin{equation}
        N_B(c) = (1 + 2^{-c})^2 \parens{\frac{\lambda t}{\epsilon}}^{1 - 1/2k} \parens{\frac{2k+1}{2k + \Upsilon}}^{1/2k} \frac{2^{1-1/k}}{\Upsilon^{1/2k}}, \label{eq:nb_parametrization}
    \end{equation}
    where $c \in \Theta(1)$ is a constant and can be negative. This specific form of $(1+2^{-c})^2$ leads to many simplifications, the first of which is $\chi = 2^{-c} \lambda L^{-1}$, following the definition of $\chi$ from Eq. \eqref{eq:prob_def}. This allows us to compute the size of the sampling set of indices $|\probIndexSet|$ as we just need to determine which $j$ leads to $1 - p_j < 1$. This is shown as
\begin{align}
    1 - p_j &< 1 \\
    \frac{\chi }{h_j} &< 1 \\
    2^j &< 2^c \frac{L}{\lambda} \\
    j &< c + \log_2 {L} - \log_2{\lambda}.
\end{align}
We next compute $\lambda = \sum_i h_i = \sum_{i} 2^{-i} = 1 - 2^{-L}$ to yield a value for $\expect{L_A}$. Since $\log (1-x) \in \bigo{x}$ for $x \to 0$, it is easy to see that $|S|$ is asymptotically small compared to $L$:
$$|S| = \lfloor j \rfloor \leq c + \log_2 (L) + \bigo{2^{-L}} \in \Theta \parens{\log_2 L}. $$
This shows that $\expect{L_A}$ satisfies it's asymptotic requirements:
\begin{align}
    \expect{L_A} &= \sum_{i \in \probIndexSet} \expect{I_i} \\ 
    &= \sum_{i \in \probIndexSet} p_i \\
    &= \sum_{i \in \probIndexSet} 1 - (1-p_i) \\
    &= |\probIndexSet| - \sum_{i \in \probIndexSet} \frac{\chi}{h_i} \\
    &= |\probIndexSet| - \frac{2^{-c} \lambda}{L} \sum_{i \in \probIndexSet} 2^{-i} \\
    &\in \Theta(|\probIndexSet|) \\
    &= \Theta(\log_2 L) \subset o(L).
\end{align}
The next task is to show that $L_A$ does not deviate from $\expect{L_A}$ in an asymptotically significant way with high probability. As $L_A$ is the sum of Bernoulli random variables we can use the multiplicative Chernoff bound
\begin{align}
    \prob{|L_A - \expect{L_A}| > \delta \expect{L_A}} &\leq 2 e^{- \delta^2 \expect{L_A} / 3} \\
    & \in \Theta \parens{\frac{1}{L^{\delta^2}}}.
\end{align}
We see that this vanishes even for constant deviations $\delta \in O(1)$ of $L_A$ from $\expect{L_A}$, implying that any asymptotically significant deviations, such as $\delta \in O(L^a)$, vanish even quicker.
 
We now move on to bounding $\expect{\lambda_B}$. The first thing we need to do is provide an upper bound on $|\probIndexSet^C|$, which is done via a lower bound on $|\probIndexSet|$. For this we can use the simple bound $\lfloor x \rfloor \geq x - 1$. This allows us to compute the required sums for $\expect{\lambda_B}$:
\begin{align}
    \expect{\lambda_B} &= \sum_{i \in \probIndexSet} (1-p_i) h_i + \sum_{i \in \probIndexSet^C} h_i \\
    &\leq \chi |\probIndexSet| + \sum_{i = c + \log_2 L - \log_2 \lambda - 1}^L 2^{-i} \\
    &\leq \frac{2^{-c} \lambda}{L} \parens{c + \log_2 L - \log_2 \lambda} + 2^{1 - c - \log_2(L) + \log_2 (\lambda) + 1} - 2^{-L} \\
    &= \frac{2^{-c} \lambda}{L}(c + \log_2 L - \log_2 \lambda) + \frac{2^{2 - c} \lambda}{L} - 2^{-L} \\
    &\leq \frac{2^{-c} \lambda}{L}(4 + c + \log_2 L - \log_2 \lambda) \\
    &\in \bigo{\frac{(1 - 2^{-L})(4 + c + \log_2 L + \sum_{k=1}^{\infty} 2^{-kL}/k )}{L}} \\
    &\subseteq \bigotilde{L^{-1}} \\
    &\subseteq o(1).
\end{align}
For the tail bound necessary on $\lambda_B$ we need to show that with high probability it will be in $o(1)$. For this a straightforward application of Markov's inequality to bound the probability that $\lambda_B$ is greater than some constant suffices:
\begin{align}
    \prob{\lambda_B \geq x} &\leq \frac{\expect{\lambda_B}}{x} \\
    &\leq \frac{2^{-c} \lambda }{x L} (4 + c + \log_2 L -\log_2 \lambda) \\ 
    &\in \bigotilde{\frac{1}{x L}},
\end{align}
which can be clearly seen to approach 1 as $L \to \infty$ for any constant value $x \in \Theta(1)$. Note that Markov's inequality is applicable due to our assumption that $h_i \geq 0$ for all $i$.

The last expression we need to satisfy is $N_B \in \Theta(L_A)$. To do so we will use the bounds on $t,\epsilon$ to find a useful expression for $N_B$
\begin{align}
    C_{QD} &= C_{Trott} \\
    \frac{4 \lambda^2 t^2}{\epsilon} &= \Upsilon^{2 + 1/2k} L \frac{t^{1 + 1/2k}}{\epsilon^{1/2k}} \alpha_{comm}(H,2k)^{1/2k} (4/2k+1)^{1/2k} \\
    \parens{\frac{\lambda t}{\epsilon}}^{1-1/2k} &= \frac{\Upsilon^{2+1/2k}}{4^{1-1/2k} (2k+1)^{1/2k}} \frac{L \alpha_{comm}(H,2k)^{1/2k}}{\lambda^{1/2k}}.
\end{align}
Now we plug this in to the parametrization of $N_B$ given in Eq. \eqref{eq:nb_parametrization}
\begin{align}
    N_B &= (1 + 2^{-c})^{2} \parens{\frac{2k+1}{2k + \Upsilon}}^{1/2k} \frac{2^{1-1/k}}{\Upsilon^{1/2k}} \frac{\Upsilon^{2+1/2k}}{4^{1-1/2k} (2k+1)^{1/2k}} \frac{L \alpha_{comm}(H,2k)^{1/2k}}{\lambda^{1/2k}} \\
    &\in \Theta \parens{\frac{L \alpha_{comm}(H,2k)^{1/2k}}{\lambda^{1/2k}}}.
\end{align}
Since $\lambda \to 1$, we can drop it from the denominator in the asymptotic limit. It is then straightforward to see that if $$\alpha_{comm}(H,2k)^{1/2k} \in \Theta\parens{\frac{\log_2 L}{L}},$$ then $N_B \in \Theta(\log_2 L) = \Theta(L_A)$, meaning that $N_B$ satisfies its asymptotic requirements from Theorem \ref{thm:higher_order_improvements_general}. This completes the proof.
\end{proof}

\section{General Composite Channels}
Now that we have gone over the details of how to apply a Composite simulation for the particular case of Trotter and QDrift, let us now consider how one could generalize this exact same idea to any combination of known simulation methods.  This approach, we will see, provides a broader perspective within which LCU, multiproduct, Trotter and QDrift can be seen to be special cases of the larger framework.

\begin{claim}
Assume a Hilbert space of the form $\mathcal{H}_A \otimes \mathcal{H}_{sys}$ where $\mathcal{H}_A$ is the Hilbert space of an ancillary system.  Specifically, let $N_t$ represent the number of different individual exponentials that are called in a single step and let $W_j$ and $V_j $ be for each $j\in 1,\ldots, N_t$ be unitary operations acting on disjoint subspace $\mathcal{H}_{A,j}$ such that $\bigcup_j \mathcal{H}_{A,j} = \mathcal{H}_A$.  Next let ${\rm Sel}_j\ket{p}_j \ket{\psi}_{sys} = \ket{p}_j e^{-i H_j t} \ket{\psi}_{sys}$ select a particular time evolution based on the index stored in the $j^{\rm th}$ register then any Composite channel consisting of randomized Trotter formulas, determistic Trotter formulas, multiproduct formulas and QDrift can be thought of as special cases of the following template
$$
\capU_{comp} : \ketbra{0}{0}\otimes\rho \mapsto {\rm Tr}_{A} \left(\prod_{j=1}^{N_t} (V_j\otimes \openone) {\rm Sel}_j (W_j\otimes \openone) \rho (W_j^\dagger\otimes \openone){\rm Sel}_j^\dagger (V_j^\dagger\otimes \openone)\right)
$$
\end{claim}
\begin{proof}
First let us show that a segment of QDrift can be thought of as a special case of this template.  In this case let us take $V_j=\openone$ and $W_j$ to be a unitary such that for a Hamiltonian of the form $H_{B} = \sum_{k=0}^{L-1} h_k H_k$ with $\sum_k h_k = \lambda_B$ and $h_k\ge 0$, then
\begin{equation}
    W_j \ket{0} = \sum_{k=0}^{N-1} \sqrt{\frac{h_k}{\lambda_B}} \ket{k}.
\end{equation}
We then have from some elementary algebra that
\begin{equation}
    {\rm Tr}_{A_j}  (V_j\otimes \openone) {\rm Sel}_j (W_j\otimes \openone) \rho (W_j^\dagger\otimes \openone){\rm Sel}_j^\dagger (V_j^\dagger\otimes \openone) = \sum_{k} \frac{h_k}{\lambda_B} e^{-i \lambda_B H_k t} \rho e^{i \lambda_B H_k t} = \qdchan(H_B, t) \circ \rho.
\end{equation}
As each Hilbert space $\mathcal{A}_j$ is assumed to be disjoint, we then have that this also holds for the multi-segment case as well.  Specifically, let us assume that for some positive integer $Q$ that the channel maps
\begin{equation}
    \ketbra{0}{0}\otimes\rho \mapsto {\rm Tr}_{A} \left(\prod_{j=1}^{Q} (V_j\otimes \openone) {\rm Sel}_j (W_j\otimes \openone) \rho (W_j^\dagger\otimes \openone){\rm Sel}_j^\dagger (V_j^\dagger\otimes \openone)\right) = \qdchan(H_B, t)^{\circ Q} \circ \rho \label{eq:base}
\end{equation}
We then have using the fact that each of the $V_j$ and $W_j$ act on disjoint Hilbert spaces that
\begin{align}
    &{\rm Tr}_{A} \left(\prod_{j=1}^{Q+1} (V_j\otimes \openone) {\rm Sel}_j (W_j\otimes \openone) \rho (W_j^\dagger\otimes \openone){\rm Sel}_j^\dagger (V_j^\dagger\otimes \openone)\right) \nonumber\\
    &= {\rm Tr}_{A_{Q+1}} {\rm Tr}_{\bigcup_{q=1}^Q A_{q}} \left(\prod_{j=1}^{Q+1} (V_j\otimes \openone) {\rm Sel}_j (W_j\otimes \openone) \rho (W_j^\dagger\otimes \openone){\rm Sel}_j^\dagger (V_j^\dagger\otimes \openone)\right)\nonumber\\
    &= {\rm Tr}_{A_{Q+1}}  \Biggr((V_{Q+1}\otimes \openone) {\rm Sel}_{Q+1} (W_{Q+1}\otimes \openone)\nonumber\\
    &\qquad \times {\rm Tr}_{\bigcup_{q=1}^Q A_{q}}\left(\prod_{j=1}^{Q} (V_j\otimes \openone) {\rm Sel}_j (W_j\otimes \openone) \rho (W_j^\dagger\otimes \openone){\rm Sel}_j^\dagger (V_j^\dagger\otimes \openone)\right)\nonumber\\
    & \qquad \times (W_{Q+1}^\dagger\otimes \openone){\rm Sel}_{Q+1}^\dagger (V_{Q+1}^\dagger\otimes \openone) \Biggr)\label{eq:induction}
\end{align}
We then have from~\eqref{eq:base} and~\eqref{eq:induction} that
\begin{equation}
    {\rm Tr}_{A} \left(\prod_{j=1}^{Q+1} (V_j\otimes \openone) {\rm Sel}_j (W_j\otimes \openone) \rho (W_j^\dagger\otimes \openone){\rm Sel}_j^\dagger (V_j^\dagger\otimes \openone)\right) = \qdchan(H_B, t)^{\circ Q+1} \circ \rho,
\end{equation}
and thus the claim holds for all positive integer $Q$ by induction since we already demonstrated the $Q=1$ case.

Next we wish to show that any randomized Trotter formula can be thought of as a special case of this framework.  While we have solely focused on deterministic Trotter formulas in this paper, we generalize to randomly ordered formulas in this section, as developed in ~\cite{childs2019faster}, because the formalism required is the same as the deterministic case. We define the first-order Trotter channel given by a permutation $\sigma \in S_L$, where $S_L$ is the symmetric group on $L$ elements, as 
\begin{equation}
    \capU_{\sigma}^{(1)}(H, t) : \rho \mapsto \prod_{i = 1}^{L} e^{i H_{\sigma(i)} t} \rho \prod_{i = L}^{1} e^{-i H_{\sigma(i)}t}.
\end{equation}
The first-order Randomized Trotter-Suzuki channel is then given as $\frac{1}{L!} \sum_{\sigma \in S_L} \capU_{\sigma}^{(1)} \circ \rho$ . Our goal is to show that this is clearly encodable within the framework developed here.

We start by noting that for a single iteration a single ancilla register is sufficient. For a Hamiltonian $H_A$ we then take $V_{RTS} = \openone$, $W_{RTS} \ket{0} = \sum_{\sigma \in S_{L_A}} \frac{1}{\sqrt{L_A!}} \ket{\sigma}$, where $\ket{\sigma}$ can be an arbitrary binary encoding of an index of $S_{L_A}$. Then define the action of the Select operator as ${\rm Sel_\sigma} \parens{\ketbra{\sigma}{\sigma} \otimes \rho} {\rm Sel_{\sigma}^\dagger} = \ketbra{\sigma}{\sigma} \otimes \capU_{\sigma}^{(1)} \circ \rho $. This then gives the overall action of the channel as
\begin{equation}
    {\rm Tr}_{A} \parens{ (V_{RTS} \otimes \openone) {\rm Sel_{\sigma}} (W_{RTS} \otimes \openone) \ketbra{0}{0} \otimes \rho (W_{RTS}^\dagger \otimes \openone) {\rm Sel_{\sigma}^\dagger} (V_{RTS}^\dagger \otimes \openone)} = \frac{1}{L_A!} \sum_{\sigma \in S_{L_A}} \capU_{\sigma}^{(1)}(H_A, t) \circ \rho,
\end{equation}
which matches the first-order randomized Trotter-Suzuki formula as given in \cite{childs2019faster}. It is a straightforward process to extend this mapping to higher-order randomized formulas.


The final case that fits within this framework is the multi-product formula.  The primary difference between this case and the probabilistic Trotter formula is that a quantum superposition of the Trotter formulas is used~\cite{childs2012hamiltonian,low2019well, faehrmann2021randomizing}.  
Specifically, this method finds a set of $N$ coefficients, $c_\sigma$, such that
\begin{equation}
\|\sum_{\sigma=1}^{N} c_\sigma  (\prod_{q=1}^{L_A} \exp(-i[H_A]_{q} t/k_{\sigma}))^{k_\sigma} - e^{-iH_A t}\| \in O(t^{2N+1})
\end{equation}
To implement multiproduct formulas within our formalism, we define $$W_{MPF}\ket{0} \coloneqq \sum_{\sigma =1}^{N} \sqrt{|c_\sigma|} / \sqrt{\sum_j |c_\sigma|} \ket{\sigma},$$ with $V_{MPF}^\dagger= W_{MPF}$.  The coefficients $c$ are chosen such that they are the solutions to the Vandermonde system of equations for some sequence of positive integers $k_j$~\cite{childs2012hamiltonian}
\begin{equation}
\begin{bmatrix}
1 & k_1^{-1} & k_1^{-2} &\cdots & k_1^{-N+1}\\
1 & k_2^{-1} & k_2^{-2} &\cdots & k_2^{-N+1}\\
\vdots& & &\ddots &\vdots\\
1 & k_N^{-1} & k_N^{-2} &\cdots & k_N^{-N+1}\\
\end{bmatrix}^T \begin{bmatrix}
c_1 \\c_2 \\ \vdots \\c_N
\end{bmatrix}
=
\begin{bmatrix}
1 \\0 \\\vdots \\0
\end{bmatrix}
\end{equation}
We take
${\rm Sel}_\sigma\ket{\sigma}\ket{\psi} = \ket{\sigma} U_{\rm TS}(H_A,t) \ket{\psi} = \ket{\sigma}(\prod_{q=1}^{L_A} \exp(-i[H_A]_{q} /k_\sigma t))^{k_\sigma}\ket{\psi}$. This  encapsulates the select operation needed in~\cite{wiebe2010higher,low2019well}.  The LCU Lemma~\cite{childs2012hamiltonian,berry2015simulating} then implies that this channel will map, 
\begin{equation}
    (\bra{0}\otimes \openone)\Lambda(\ketbra{0}{0}\otimes \rho)(\ket{0}\otimes \openone) = \kappa^{-2}\sum_{\sigma' =1}^N\sum_{\sigma=1}^{N} c_\sigma  (\prod_{q=1}^{L_A} \exp(-i[H_A]_{q} t/k_{\sigma}))^{k_\sigma}\rho c_{\sigma'}  (\prod_{q=1}^{L_A} \exp(-i[H_A]_{q} t/k_{\sigma}))^{-k_\sigma},
\end{equation}
where $\kappa^{-1}$ is a constant needed to block encode the formula within a larger unitary matrix.  Thus multiproduct formulas can also be considered as special cases of a Composite channel.

Finally, following the exact same arguments used above since the registers used to block encode each of these individual channels are taken to be disjoint, any composition of the channels applied in this manner forms a Composite channel as we defined it above.  This proves our claim.
\end{proof}
Note that above we do not discuss the complications of using oblivious amplitude amplification to prevent the success probability of multiproduct formulas shrinking exponentially with the number of segments~\cite{berry2015hamiltonian}.  This generalization is straightforward though and simply requires wrapping the evolution within a larger unitary transformation that resembles Grover's search.  Further, it is also clear that the exact same ideas used for multi-product formulas could be used to implement a linear combinations of unitaries approach such as~\cite{berry2015simulating} or~\cite{low2019hamiltonian}.  For simplicity, we choose our language above to focus on the case of composite Trotter-like channels but this approach also falls within our framework. \cite{Chen_2021} Qubitization is somewhat of an awkward fit within this approach~\cite{low2018hamiltonian}, as it involves a sequence of rotations as opposed to simple linear expressions, and so it remains the one major simulation method that falls outside of the generalization of the Composite channel that we present above. We hope that this provides a compact framework for future analysis of composite simulations with more involved simulation techniques.

\section{Discussion}\label{sec:discussion}
This work has presented a new framework to analyze and design quantum simulation algorithms that have a compositional character.  In particular, we examine mixing Trotter-Suzuki methods and QDrift and find evidence that the whole is greater than the sum of its parts.  The central new concept behind this work is the idea of term partitioning, both in a deterministic as well as a randomized setting, which allows us to selectively apply a simulation algorithms to the parts of the Hamiltonian that benefit from high accuracy while relegating less sensitive parts of the Hamiltonian to a low-accuracy approximation.  

We show for low order formulas that in cases where there are a small number of large terms that are (mostly) mutually commuting and a large number of smaller terms in the Hamiltonian, then a deterministic partitioning exists such that the lowest order Trotter formula combined with QDrift provides an asymptotic advantage over either model.  For higher-order Trotter formulas we provide a set of asymptotic regimes for parameters of a predetermined partition, namely the number of terms in the Trotter partition and the sum of the spectral norms for the QDrift partition, that yield asymptotic improvements over Trotter or QDrift individually. Moreover we provide a method for producing partitions probabilistically that is easy to compute and exactly matches the operator exponential cost for second-order Trotter formulas and QDrift in the relevant limits. 

To showcase how each of these results can be used to analyze a Hamiltonian we specifically examined the case where a given $H$ has exponentially decaying spectral norms ($h_i = 2^{-i}$) and a small commutator structure ($\alpha_{comm}(H; 2k) \in \Theta (\log_2 L / L)$). We show that by using the probabilistic scheme provided that we can produce partitions that satisfy the asymptotic requirements in Theorem \ref{thm:higher_order_improvements_general} with probability approaching 1 as $L \to \infty$. We note that although this example Hamiltonian is contrived for the purposes of demonstrating our techniques there are many Hamiltonians, such as electronic structure Hamiltonians, that are dominated by a handful of large terms with many more  small terms. Commutator structures are highly dependent on specific Hamiltonians but our techniques give an asymptotic scaling that leads to \emph{provable} advantage. This gives strong evidence that optimized partitions that take advantage of information about the Hamiltonian could lead to significant empirical improvements.

Looking forward, while this framework shows a way that we can think about combining several disparate Hamiltonian simulation methods, this contribution is by no means the end of this line of inquiry.  Firstly, this work was inspired to no small extent by the coalescing scheme presented in~\cite{coalescing_con_wiebe} wherein a factor of $10$ improvement was found numerically by decreasing the frequency with which low importance Hamiltonian terms are applied when looking at chemical simulations.  Our framework provides a much more general way of thinking about such partitioning and further numerical work would be useful to understand the impact that these ideas will have on systems of interest, such as chemical compounds or lattice gauge theories.

More broad approaches within our framework can also be considered specifically, multiple orders of Trotter formulas could be considered and further randomization over the ordering of terms can also be performed easily within the randomization part of our algorithm.  Incorporating further gradations in the accuracy of the underlying formula could lead to further practical improvements and pave the way for more highly optimized compilers. In systems where commutator or spectral norm information about subsets of the Hamiltonian is known these techniques could bring significant circuit depth reductions for simulations on quantum computers.

\acknowledgements{
M.H. and N.W. were supported by the US DOE National Quantum Information Science Research Centers, Co-design Center for Quantum Advantage (C2QA) under contract number DE-SC0012704.  We thank Burak Catli for useful conversations and feedback on this work.  
}

\bibliographystyle{quantum}
\bibliography{bib}

\appendix
\section{Moment Bounds for Higher-Order Formulas} \label{sec:appendix_a}
We now return to proving the moment bounds from Section \ref{sec:higher_order_improvements}. 
\laSquared*
\begin{proof}
Given that the simplest definition of our probabilities is for $1-p_i$ we will try to work with expressions for the $B$ channel as much as possible. It is easy to convert between the two as
\begin{equation}
    \expect{L_A^2} = \expect{(L - L_B)^2} = L^2 -2L\expect{L_B} + \expect{L_B^2}.
\end{equation}
The expectation value of $L_B$ then follows from plugging in the definitions
\begin{equation}
    \expect{L_B} = \sum_i\expect{I_i^B} = \sum_i 1-p_i = \chi \sum_{i \in \probIndexSet} \frac{1}{h_i} + |\probIndexSet^C|. \label{eq:subPartLAsquare}
\end{equation}

Now we find a relatively simple upper bound for $\expect{L_B^2}$ if use the two facts that $I_i^B$ and $I_j^B$ are independent for $i \neq j$ and that $\parens{I_i^B}^2 = I_i^B$
\begin{align}
    \expect{L_B^2} &= \expect{\parens{\sum_i I_j^B}^2} \\
    &= \sum_i \expect{I_i^B} + \sum_{i \neq j} \expect{I_i^B}\expect{I_j^B} \\
    &= \sum_i (1-p_i) + \parens{\sum_i (1-p_i)}^2 - \sum_i (1-p_i)^2\\
    &= \chi \sum_{i \in \probIndexSet} \frac{1}{h_i} +|\probIndexSet^C| + \parens{\chi \sum_{i \in \probIndexSet} \frac{1}{h_i} +|\probIndexSet^C|}^2 - \sum_{i \in \probIndexSet} \frac{\chi^2}{h_i^2} - |\probIndexSet^C|. \label{eq:expectLBsquare}
\end{align}
Combining equations \eqref{eq:subPartLAsquare} and \eqref{eq:expectLBsquare} we get the following expression for an upper bound on $L_A^2$
\begin{align}
    \expect{L_A^2} &= L^2 - 2L(\chi \sum_{i \in \probIndexSet}\frac{1}{h_i} + |\probIndexSet^C|) + \chi \sum_{i \in \probIndexSet}\frac{1}{h_i} +\parens{\chi \sum_{i \in \probIndexSet}\frac{1}{h_i} + |\probIndexSet^C|}^2 - \sum_{i \in \probIndexSet} \frac{\chi^2}{h_i^2} \\
    &= L^2 -2L|\probIndexSet^C| + |\probIndexSet^C|^2 +  \parens{-2L + 1 + 2 |\probIndexSet^C|}\sum_{i \in \probIndexSet}\frac{\chi}{h_i} + \parens{\sum_{i \in \probIndexSet} \frac{\chi}{h_i}}^2 - \sum_{i \in \probIndexSet} \frac{\chi^2}{h_i^2} \\
    &= \parens{L - |\probIndexSet^C|}^2 +  \parens{1 - 2 |\probIndexSet|}\sum_{i \in \probIndexSet} \frac{\chi}{h_i} + \sum_{i \neq j \in \probIndexSet^2} \frac{\chi}{h_i} \frac{\chi}{h_j} \\
    &\leq |\probIndexSet|^2 + \parens{1 - 2|\probIndexSet|} \sum_{i \in \probIndexSet}\frac{\chi}{h_i} + \parens{|\probIndexSet| - 1} \sum_{i \in \probIndexSet} \frac{\chi}{h_i} \\
    &= |\probIndexSet|^2 - |\probIndexSet| \sum_{i \in \probIndexSet} \frac{\chi}{h_i}. \label{eq:LAasymptotic}
\end{align}
Note that we only used the following inequality $\frac{\chi}{h_i} = 1 - p_i \leq 1$ for $i \in \probIndexSet$. 
\end{proof}

\qUpperBounds*
\begin{proof}
$\expect{Q(t)}$ can be computed very similarly to $L_A$ above, since $\expect{Q(t)} \propto \expect{\lambda_B^2}$. This second moment for $\lambda_B$ mostly follows from the definitions but we first will get an easier expression involving the indicator variables
\begin{align}
    \expect{\lambda_B^2} &= \expect{\parens{\sum_i h_i I_i^B}^2} \\
    &= \sum_i h_i^2\expect{ I_i^B}  + \sum_{i \neq j} h_i h_j \expect{ I_i^B I_j^B} \\
    &= \sum_i h_i^2\expect{ I_i^B}  + \sum_{i \neq j} h_i \expect{ I_i^B} h_j \expect{I_j^B} \\
    &= \sum_i h_i^2\expect{ I_i^B}  + \parens{\sum_{i} h_i \expect{ I_i^B}}^2 - \sum_j h_j^2 \expect{I_j^B}^2,
\end{align}
where we used the fact that $I_i^B$ is independent from $I_j^B$ for all $i \neq j$. Now we can utilize our probability distributions as $\expect{I_i^B} = 1-p_i$, which yields
\begin{align}
    \expect{\lambda_B^2} &= \sum_{i \in \probIndexSet} h_i \chi + \sum_{i \in \probIndexSet^C} h_i^2 + \parens{\chi |\probIndexSet| + \lambda_{\probIndexSet^C}}^2 - \sum_{j \in \probIndexSet} h_j^2 (1-p_j)^2 - \sum_{j \in \probIndexSet^C} h_j^2 (1-p_j)^2 \\
    &= \lambda_{\probIndexSet} \chi + \sum_{i \in \probIndexSet^C} h_i^2 + \chi^2 |\probIndexSet|^2 + 2 \chi |\probIndexSet| \lambda_{\probIndexSet^C} + \lambda_{\probIndexSet^C}^2 - \chi^2 |\probIndexSet| - \sum_{j \in \probIndexSet^C} h_j^2 \\
    &\leq \chi^2 |\probIndexSet|^2 + \chi \parens{\lambda_{\probIndexSet} + 2 \lambda_{\probIndexSet^C} |\probIndexSet|} + \lambda_{\probIndexSet^C}^2 \\
    &=\chi \lambda_{\probIndexSet}  + \parens{\chi |\probIndexSet| + \lambda_{\probIndexSet^C}}^2. \label{eq:lambdaBsquared}
\end{align}
The only inequality comes from dropping the correction term $\chi^2 |\probIndexSet|$, which is subleading to $\chi^2 |\probIndexSet|^2$. 

Our expression for a upper bound on $\expect{Q(t)}$ is then proportional to the above expression \eqref{eq:lambdaBsquared}
\begin{equation}
    \expect{Q(t)} \leq \frac{2t^2}{N_B} \parens{\chi \lambda_{\probIndexSet}  + \parens{\chi |\probIndexSet| + \lambda_{\probIndexSet^C}}^2}.
\end{equation}
To compute $\expect{Q(t)^2}$ we can reduce it to our prior results. Since $\expect{Q(t)^2} = \frac{4 t^4}{N_B^2} \expect{\lambda_B^4}$, we will use the following upper bound
\begin{equation}
    \expect{\lambda_B^4} = \parens{\sum_i h_i I_i^B}^4 \leq \parens{\sum_i h_i 1}^2 \parens{\sum_i h_i I_i^B}^2 = \lambda^2 \expect{\lambda_B^2}.
\end{equation}
This means we can re-use the above computation as 
\begin{equation}
    \expect{Q(t)^2} \leq \frac{2t^2 \lambda^2}{N_B}\expect{Q(t)}. \label{eq:QsquaredAsymptotic}
\end{equation}
\end{proof}

\pUpperBound*
\begin{proof}
To simplify this we will use intermediate steps from the calculation of $P_{\max}(t)$ that bound the $\alpha_{comm}$ factors, namely equations \eqref{eq:alphaCommA}, \eqref{eq:alphaCommAB} which are repeated below
\begin{align}
    \alpha_{comm}(A, 2k) &\leq 2^{2k} \lambda_A^{2k+1} \\
    \alpha_{comm}(\set{A,B},2k) &\leq 2^{2k} \sum_{l=1}^{2k} \lambda_A^{l} \lambda_B^{2k+1 - l} \\
    P(t) &=  \frac{2^2 (t \Upsilon)^{2k + 1}}{2k+1} (\Upsilon \alpha_{comm}(A, 2k) + \alpha_{comm}(\set{A,B}, 2k))
\end{align}
We will use these to compute a useful upper bound on $\expect{P(t)}$. Since our random variables are more easily described for the $I_i^B$ variables, we will convert all powers of $\lambda_A$ into functions of $L$ and $\lambda_B$ as well as upper bound both by $\lambda_A$ and $\lambda_B$ by $\lambda$. This results in an upper bound on $\expect{P(t)}$ as
\begin{equation}
    \expect{P(t)} \leq \frac{2^{2 + 2k} (t \Upsilon)^{2k + 1}}{2k+1} \expect{\Upsilon \lambda_A^{2k+1} + \sum_{l=1}^{2k} \lambda_A^l \lambda_B^{2k+1-l}}. \label{eq:expectPbad}
\end{equation}
We will simplify the expectation value using the facts that $\lambda_A, \lambda_B \leq \lambda$ and that in $\sum_{l = 1}^{2k} \lambda_A^{l} \lambda_B^{2k + 1 - l}$ each term has at least one factor of $\lambda_A \lambda_B$. This results in the following simplifications
\begin{align}
    \expect{\Upsilon \lambda_A^{2k+1} + \sum_{l=1}^{2k} \lambda_A^l \lambda_B^{2k+1-l}} &\leq \lambda^{2k-1} \expect{\Upsilon \lambda_A^2 + (2k) \lambda_A \lambda_B} \\
    &\leq \Upsilon \lambda^{2k-1} \expect{\lambda_A^2 + \lambda_A \lambda_B} \\
    &= \Upsilon \lambda^{2k-1} \expect{\lambda_A(\lambda_A + \lambda_B)} \\
    &= \Upsilon \lambda^{2k} \parens{\lambda - \expect{\lambda_B}} \\
    &= \Upsilon \lambda^{2k} \parens{\lambda_{\probIndexSet} - \chi |\probIndexSet|}, \label{eq:exactUpperP}
\end{align}
where we used an exact expression for $\expect{\lambda_B}$ that is a straightforward computation in addition to the fact that $2k \leq \Upsilon$ for $k \geq 1$. Combining the above expressions \eqref{eq:exactUpperP} and \eqref{eq:expectPbad} our final expression for an upper bound on $\expect{P(t)}$ is 
\begin{equation}
    \expect{P(t)} \leq \frac{(2 \Upsilon)^{2 + 2k}}{2k+1} t^{2k+1} \lambda^{2k} \parens{\lambda_{\probIndexSet} - \chi |\probIndexSet|}
\end{equation}
\end{proof}

\end{document}